\documentclass{LMCS}

\def\doi{8 (2:05) 2012}
\lmcsheading%
{\doi}
{1--30}
{}
{}
{Sep.~17, 2011}
{May.~31, 2012}
{}

\pdfoutput=1
\usepackage{textcomp}
\usepackage{amsfonts}
\usepackage{amssymb}
\usepackage[leqno]{amsmath}
\usepackage{stmaryrd}
\usepackage[dvipsnames,usenames]{color}
\usepackage{paralist}
\usepackage{xspace}
\usepackage{mismacros}
\usepackage{graphicx}
\usepackage{url}
\usepackage{proof, mathpartir}
\usepackage{syntax}
\usepackage{fixltx2e}
\usepackage{bussproofs}
\usepackage{subfig}
\usepackage{version}
    \includeversion{TR}
    \excludeversion{CA}
\usepackage{hyperref,enumerate}

\usepackage{listings}
\newcommand{\ttt}[1]{$\mathtt{#1}$}

\newcommand{\tobeupdated}[1]{{\color{red}#1}\marginpar{\tiny\textsc{To Be Updated}}}
\newcommand{\upforinclusion}[1]{{\color{blue}#1}\marginpar{\tiny\textsc{Up For Inclusion}}}

\newcommand{\dpnote}[1]{\marginpar{\color{magenta}David: #1}}
\newcommand{\ftk}[1]{\marginpar{\textcolor{Green}{#1}}}

\begin{document}

\title{Secure the Clones}

\author [T.~Jensen]{Thomas Jensen}
\author [F.~Kirchner]{Florent Kirchner}
\author [D.~Pichardie]{David Pichardie}
\address{INRIA Rennes -- Bretagne Atlantique, France}
\email  {firstname.lastname@inria.fr}
\thanks {This work was supported in part by the ANSSI, the ANR, and the
         \emph{R\'egion Bretagne}, respectively under the Javasec, Parsec, and
         Certlogs projects.}



\begin{abstract}
  Exchanging mutable data objects with untrusted code is a delicate
  matter because of the risk of creating a data space that is
  accessible by an attacker. Consequently, secure programming guidelines for
  Java stress the importance of using defensive copying before accepting or
  handing out references to an internal mutable object.
  However, implementation of a copy method (like clone()) is entirely
  left to the programmer. It may not provide a sufficiently deep copy of
  an object and is subject to overriding by a malicious sub-class. Currently 
  no language-based mechanism supports secure object cloning. 
  This paper proposes a type-based annotation system for defining modular
  copy policies for class-based object-oriented programs. 
  A copy policy specifies the maximally allowed sharing between an
  object and its clone. We present a static
  enforcement mechanism that will guarantee that all classes fulfil their
  copy policy, even in the presence of overriding of copy methods, and  
  establish the semantic correctness of the overall approach in Coq.
  The mechanism has been implemented and experimentally evaluated on
  clone methods from several Java libraries.
\end{abstract}

\maketitle

\section{Introduction}

Exchanging data objects with untrusted code is a delicate matter because
of the risk of creating a data space that is accessible by an attacker.
Consequently, secure programming guidelines for Java such as those proposed by
Sun \cite{SunGuidelines:2010} and CERT \cite{CertGuidelines:2010} stress the
importance of using defensive \emph{copying} or \emph{cloning} before
accepting or handing out references to an internal mutable object. There are
two aspects of the problem:
\begin{enumerate}[(1)]
\item If the result of a method is a reference to an internal mutable object,
  then the receiving code may modify the internal state. Therefore, it is
  recommended to make copies of mutable objects that are returned as
  results, unless the intention is to share state. 
\item If an argument to a method is a reference to an object coming
  from hostile code, a local copy of the object should be
  created. Otherwise, the hostile code may be able to modify the internal
  state of the object.
\end{enumerate}

\noindent A common way for a class to provide facilities for copying objects is
to implement a \ttt{clone()} method that overrides the cloning
method provided by \ttt{java.lang.Object}. 
The following code snippet, taken from Sun's Secure Coding
Guidelines for Java, demonstrates how a \ttt{date} object is cloned
before being returned to a caller:
\begin{lstlisting}
public class CopyOutput {
    private final java.util.Date date;
    ...
    public java.util.Date getDate() {
        return (java.util.Date)date.clone(); }
}
\end{lstlisting}

\noindent However, relying on calling a polymorphic \ttt{clone} method
to ensure secure copying of objects may prove insufficient, for two
reasons. First, the implementation of the \ttt{clone()} method is
entirely left to the programmer and there is no way to enforce that an
untrusted implementation provides a sufficiently \emph{deep} copy of
the object. It is free to leave references to parts of the original
object being copied in the new object. Second,
even if the current \ttt{clone()} method works properly, sub-classes may
override the \ttt{clone()} method and replace it with a method that does not
create a sufficiently deep clone.  For the above example to behave correctly,
an additional class invariant is required, ensuring that the \ttt{date} field
always contains an object that is of class \ttt{Date} and not one of its 
sub-classes. To quote from the CERT guidelines for secure Java programming:
\emph{``Do not carry out defensive copying using the clone() method in
constructors, when the (non-system) class can be subclassed by untrusted code.
This will limit the malicious code from returning a crafted object when the
object's clone() method is invoked.''} Clearly, we are faced with a situation
where basic object-oriented software engineering principles (sub-classing and
overriding) are at odds with security concerns.
To reconcile these two aspects in a manner that provides semantically
well-founded guarantees of the resulting code, this paper proposes a formalism
for defining \emph{cloning policies} by annotating classes and specific copy
methods, and a static enforcement mechanism that will guarantee that all
classes of an application adhere to the copy policy. Intuitively,
policies impose non-sharing constraints between the structure 
referenced by a field of an object and the structure returned by the
cloning method.  Notice, that we do not enforce that
a copy method will always return a target object that is functionally
equivalent to its source. Nor does our method prevent a sub-class from
making a copy of a structure using new fields that are not governed by
the declared policy. For a more detailed example of these limitations,
see Section~\ref{sec:pol:limitations}. 


\subsection{Cloning of Objects}
\label{sec:intro:cloning}

For objects in Java to be cloneable, their class must implement the
empty interface \ttt{Cloneable}. A default \ttt{clone} method is
provided by the class \ttt{Object}: when invoked on an object of a
class, \ttt{Object.clone} will create a new object of that class and
copy the content of each field of the original object into the new
object. The object and its clone share all sub-structures of the
object; such a copy is called \textit{shallow}. 

It is common for cloneable classes to override the default clone
method and provide their own implementation. For a generic
\ttt{List} class, this could be done as follows:

\begin{lstlisting}
public class List<V> implements Cloneable
{
    public V value;
    public List<V> next;

    public List(V val, List<V> next) {
	this.value = val;
	this.next = next; }

    public List<V> clone() {
        return new List(value,(next==null)?null:next.clone()); }
}  
\end{lstlisting}
Notice that this cloning method performs a shallow copy of the list,
duplicating the spine but sharing all the elements between the list and its
clone. Because this amount of sharing may not be desirable (for the reasons
mentioned above), the programmer is free to implement other versions of
\ttt{clone()}. For example, another way of cloning a list is by copying both
the list spine and its elements\footnote{To be type-checked by the Java
compiler it is necessary to add a cast before calling \ttt{clone()} on
\ttt{value}. A cast to a sub interface of \ttt{Cloneable} that declares a
\ttt{clone()} method is necessary.}, creating what is known as a \textit{deep}
copy. 
\begin{lstlisting}
public List<V> deepClone() {
  return new List((V) value.clone(),
                  (next==null ? null : next.deepClone())); }  
\end{lstlisting}

\noindent A general programming pattern for methods that clone objects works by first
creating a shallow copy of the object by calling the \ttt{super.clone()}
method, and then modifying certain fields to reference new copies of the original
content. This is illustrated in the following snippet, taken from the class
\ttt{LinkedList} in Fig.~\ref{fig:linkedlist}:

\begin{lstlisting}
public Object clone() {  ...
  clone = super.clone();  ...
  clone.header = new Entry<E>(null, null, null); ...
  return clone;}
\end{lstlisting}

\noindent
There are two observations to be made about the analysis of such methods.
First, an analysis that tracks the depth of the clone being returned will have
to be flow-sensitive, as the method starts out with a shallow copy that is
gradually being made deeper. This makes the analysis more costly. Second,
there is no need to track precisely modifications made to parts of the memory
that are not local to the clone method, as clone methods are primarily
concerned with manipulating memory that they allocate themselves. This
will have a strong impact on the design choices of our analysis.


\subsection{Copy Policies}
\label{sec:intro:pol}

The first contribution of the paper is a proposal for a set of semantically
well-defined program annotations, whose purpose is to enable the expression of
policies for secure copying of objects. Introducing a copy policy language
enables class developers to state explicitly the intended behaviour of copy
methods. 
%
In the basic form of the copy policy formalism, fields of classes are
annotated with \ttt{@Shallow} and \ttt{@Deep}. Intuitively, the annotation
\ttt{@Shallow} indicates that the field is referencing an object, parts of which
may be referenced from elsewhere. The annotation \ttt{@Deep}(\ttt{X}) on a
field \ttt{f} means
that 
\begin{inparaenum}[\itshape a\upshape)]
\item upon return from \ttt{clone()}, the object referenced by this field \ttt{f} is not referenced
  from elsewhere, and
\item the field \ttt{f} is copied according to the copy policy identified by \ttt{X}. 
\end{inparaenum} 
Here, \ttt{X} is either the name of a specific policy or if omitted,
it designates the default policy of the class of the field. 
%
%
For example, the following annotations:
\begin{lstlisting}
  class List  { @Shallow V value;  @Deep List next;  ...}
\end{lstlisting}
specifies a default policy for the class \ttt{List} where the \ttt{next} field
points to a list object that also respects the default copy policy for
lists. Any method in the \ttt{List} class, labelled with the \ttt{@Copy}
annotation, is meant to respect this default policy.

In addition it is possible to define other copy policies and annotate specific
\emph{copy methods} (identified by the annotation \ttt{@Copy(...)}) with the name
of these policies.
For example, the annotation\footnote{Our implementation uses a sightly
different policy declaration syntax because of the limitations imposed by the
Java annotation language.}
\begin{lstlisting}
DL: { @Deep V value; @Deep(DL) List next;};
@Copy(DL) List<V> deepClone() {
  return new List((V) value.clone(),
                  (next==null ? null : next.deepClone())); }  
\end{lstlisting}
can be used to specify a list-copying method that also ensures that
the \ttt{value} fields of a list of objects are copied according to
the copy policy of their class (which  is a stronger policy than
that imposed by the annotations of the class \ttt{List}). We give a
formal definition of the policy annotation language in
Section~\ref{section-annotations}. 

The annotations are meant to ensure a certain degree of non-sharing between
the original object being copied and its clone. We want to state
explicitly that the parts of the clone that can be accessed via fields marked \ttt{@Deep}
are unreachable from any part of the heap that was accessible before the call
to \ttt{clone()}. To make this intention precise, we provide a formal
semantics of a simple programming language extended with policy annotations
and define what it means for a program to respect a policy
(Section~\ref{sec:policysemantics}). 

\subsection{Enforcement} 
\label{sec:intro:typ}

The second major contribution of this work is to make the developer's intent,
expressed by copy policies, statically enforceable using a type system.  We
formalize this enforcement mechanism by giving an interpretation of the policy
language in which annotations are translated into graph-shaped type
structures. For example, the default annotations of the \ttt{List} class defined above
will be translated into the graph that is depicted to the right in 
Fig.~\ref{fig:listab}  (\ttt{res} is the name given to
the result of the copy method). The left part shows the concrete heap structure.

Unlike general purpose shape analysis, we take into account the programming
methodologies and practice for copy methods, and design a type system
specifically tailored to the enforcement of copy policies.  This means that
the underlying analysis must be able to track precisely all modifications to
objects that the copy method allocates itself (directly or indirectly) in a
flow-sensitive manner.  Conversely, as copy methods should not modify
non-local objects, the analysis will be designed to be more approximate when
tracking objects external to the method under analysis, and the type system
will accordingly refuse methods that attempt such non-local modifications. As
a further design choice, the annotations are required to be verifiable
modularly on a class-by-class basis without having to perform an analysis of
the entire code base, and at a reasonable cost.

\begin{figure}
  \centering
  \includegraphics[width=.8\linewidth]{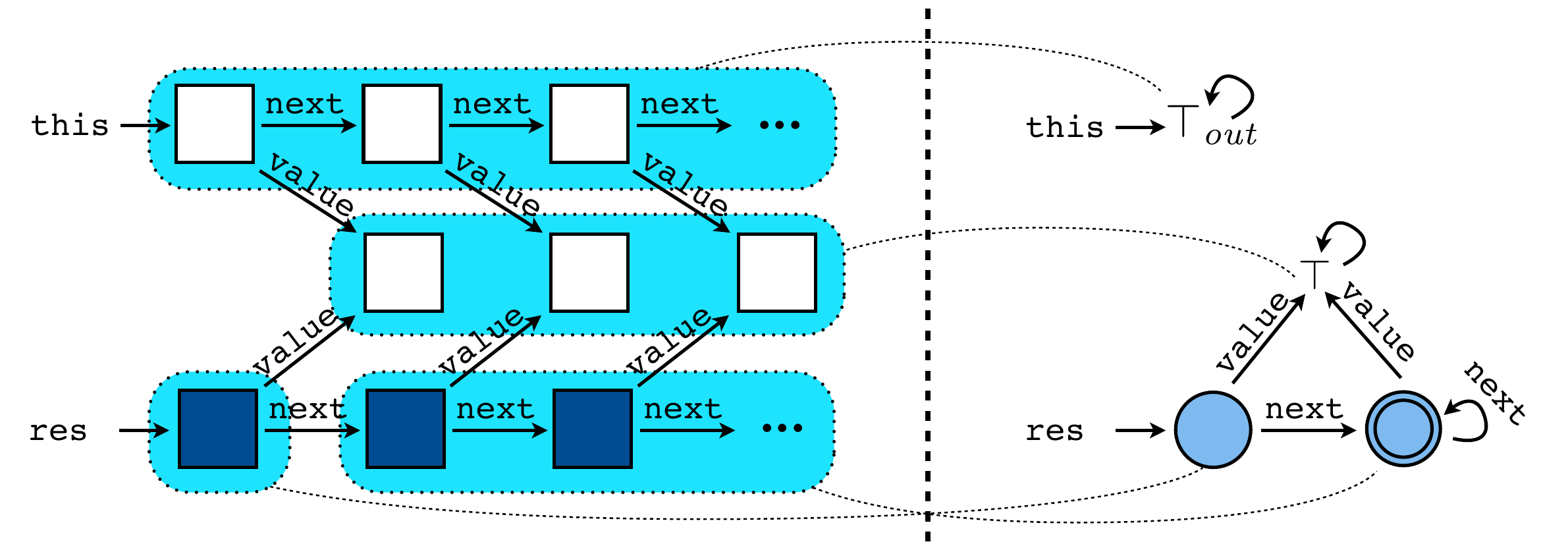}
  \caption{A linked structure (left part) and its abstraction (right part).}
  \label{fig:listab}
\end{figure}

As depicted in Fig.~\ref{fig:listab}, concrete memory cells
are either abstracted as
\begin{inparaenum}[\itshape a\upshape)]
\item $\topout$ when they are not allocated in the copy method itself (or its callee);
\item $\top$ when they are just marked as \emph{maybe-shared}; and
\item circle nodes of a deterministic 
graph when they are locally allocated and not shared. A single 
circle furthermore expresses a singleton concretization.
\end{inparaenum} 
In this example, the abstract heap representation matches the graph
interpretation of annotations, which means that the instruction set that
produced this heap state satisfies the specified copy policy.

Technically, the intra-procedural component of our analysis corresponds to
heap shape analysis with the particular type of graphs that we have defined.
Operations involving non-local parts of the heap are rapidly discarded.
Inter-procedural analysis uses the signatures of copy methods provided by the
programmer. Inheritance is dealt with by stipulating that inherited fields
retain their ``shallow/deep'' annotations.  Redefinition of a method must
respect the same copy policy and other copy methods can be added to a
sub-class. 
The detailed definition of the analysis, presented as a set of type inference
rules, is given in Section~\ref{sec:type-system}. 

This article is an extended version of a paper presented at ESOP'11~\cite{JensenKP:Esop11}.
We have taken advantage of the extra space to provide improved and more detailed
explanations, in particular of the inference mechanism and of what is
exactly is being enforced by our copy policies. We have also added
details of the proof of correctness of the enforcement
mechanism. The formalism of copy policies and the correctness theorem
for the core language defined in Section~\ref{section-annotations} have been
implemented and verified mechanically in Coq~\cite{clone-webpage}.
The added details about the proofs should especially facilitate the 
understanding of this Coq development

\section{Language and Copy Policies}\label{section-annotations}

\begin{figure}
  \centering
  \begin{small}
\begin{frameit}
\[
\var{x},\var{y} \in \Var \qquad
\var{f}         \in \Field \qquad
\var{m}         \in \Meth \qquad
\var{cn}         \in \ClassName \qquad
 X   \in \Polid
\]
\[
\begin{array}{rrrrl}\\
p&\in&\prog  & \dasig & \overline{cl} \\
\mathit{cl}&\in&\Class  & \dasig & \classkeyword~cn~[\extendskeyword~cn]~\{ \overline{\cp}~ \overline{\md} \} \\
\cp &\in& \CopySignature &\dasig& X : \{\tau\} \\
 \tau&\in&\Pol  & \dasig & \overline{(X,f)} \\
\md&\in&\MethDecl  & \dasig & \Copy(X)~\HAssign{m(x)}{}c \\
c&\in&\comm  & \dasig &\HAssign{x}{y}
 	\mid \HAssign{x}{y.f}
 	\mid \HAssign{x.f}{y} 
        \mid \HAssign{x}{\snull}\\
        &&& & \mid \HNew{x}{cn} 
        \mid \HCall{x}{\var{cn}:X}{y} \mid \UnkownCall{x}{y}  \mid \HReturn{x} \\
&&& & \mid c;c 
        \mid \HIf{c}{c} 
	\mid \HWhile{c}\\ 
\end{array}
\]
\\[2ex]
\begin{minipage}[t]{.9965\linewidth}
  \textbf{Notations:} We write $\preceq$ for the reflexive transitive closure
  of the subclass relation induced by a (well-formed) program that is fixed in
  the rest of the paper. We write $\overline{x}$ a sequence of syntactic
  elements of form $x$.
\end{minipage}
\end{frameit}
\end{small}
  \caption{Language Syntax.}
  \label{fig:syntax}
\end{figure}

The formalism is developed for a small, imperative language extended with
basic, class-based object-oriented features for object allocation, field
access and assignment, and method invocation. A program is a collection of
classes, organized into a tree-structured class hierarchy via the
$\extendskeyword$ relation. A class consists of a series of copy method
declarations with each its own policy $X$, its name $m$, its formal parameter $x$
and commands $c$ to execute. A sub-class inherits the copy methods of its
super-class and can re-define a copy method defined in one of its
super-classes. We only consider copy methods. Private methods (or static
methods of the current class) are inlined by the type checker. Other method
calls (to virtual methods) are modeled by a special instruction
$\UnkownCall{x}{y}$ that assigns an arbitrary value to $\var{x}$ and possibly
modifies all heap cells reachable from $\var{y}$ (except itself). 
The other commands are standard.
The copy method call $\HCall{x}{\var{cn}:X}{y}$ is a virtual call. The method
to be called is the copy method of name $m$ defined or inherited by the
(dynamic) class of the object stored in variable $y$. The subscript annotation
$\var{cn{:}X}$ is used as a static constraint.  It is supposed that the type
of $y$ is guaranteed to be a sub-class of class $cn$ and that $cn$ defines a
method $m$ with a copy policy $X$.  This is ensured by standard bytecode
verification and method resolution. 

We suppose given a set of policy identifiers $\Polid$, ranged over by $X$. A
copy policy declaration 
has the form $  X : \{\tau\}$ where $X$ is the identifier of the policy
signature and $\tau$ is a policy. The policy $\tau$ consists of a set of
field annotations $(X,f) ~; ~\ldots$ where $f$ is a \emph{deep} field that
should reference an object which can only be accessed via the returned pointer
of the copy method and which respects the copy policy identified by $X$. 
The use of policy identifiers makes it possible to write recursive definitions
of copy policies, necessary for describing copy properties of recursive
structures.
Any other field is implicitly \emph{shallow}, meaning that no copy properties are
guaranteed for the object referenced by the field. No further copy
properties are given for the sub-structure starting at \emph{shallow} fields.
For instance, the default copy policy declaration of the class \ttt{List} presented in
Sec.~\ref{sec:intro:pol} writes: $\mathtt{List.default}:~\{ (\mathtt{List.default},\mathtt{next}) \}$.

We assume that for a given program, all copy policies have been grouped
together in a finite map $\Pi_p:\Polid\to\Pol$.  In the rest of the paper, we
assume this map is complete, \emph{i.e.} each policy name $X$ that appears in
an annotation is bound to a unique policy in the program $p$.

The semantic model of the language defined here is store-based: $$
\begin{small} \begin{array}{rrlcl} l &\in& \Loc \\ v &\in& \Val &=& \Loc \cup
\{ \vnull \} \\ \rho &\in& \Env &=& \Var \to \Val \\ o &\in& \Object &=&
\Field \to \Val \\ h &\in& \Heap &=& \Loc \ptofin \left( \ClassName \times
\Object\right) \\ \st{\rho,h,A}&\in& \State &=&
\Env\times\Heap\times\Power(\Loc) \end{array} \end{small} $$
A program state consists of an environment $\rho$ of local variables,
a store $h$ of locations mapping\footnote{We note $\ptofin$ for
  partial functions on finite domains.} to objects in a heap and a set
$A$ of \emph{locally allocated locations}, \emph{i.e.}, the locations
that have been allocated by the current method invocation or by one of its
callees. This last component does not influence the semantic
transitions: it is used to express the type system interpretation
defined in Sec.~\ref{sec:type-system}, but is not used in the final
soundness theorem.  Each object is modeled in turn as a pair composed
with its dynamic class and a finite function from field names to values
(references or the specific \vnull\ reference for null values). We do
not deal with base values such as integers because their immutable
values are irrelevant here.
\begin{figure*}
{\centering\scriptsize
$$
\begin{array}{c}
\inferrule
 {~~}
 {\left(\HAssign{x}{y},\stb{\rho,h,A}\right) \leadsto \str{\rho[x\mapsto\rho(y)],h,A}}
\qquad
\inferrule
 {~~}
 {\left(\HAssign{x}{\snull},\stb{\rho,h,A}\right) \leadsto \str{\rho[x\mapsto\vnull],h,A}}
\\[3ex]
\inferrule
 {\rho(y) \in \dom(h)}
 {\left(\HAssign{x}{y.f},\stb{\rho,h,A}\right) \leadsto \str{\rho[x\mapsto h(\rho(y),f)],h,A}}
\qquad
\inferrule
 {\rho(x) \in \dom(h)}
 {\left(\HAssign{x.f}{y},\stb{\rho,h,A}\right) \leadsto \str{\rho,h[(\rho(x),f)\mapsto \rho(y)],A}}
\\[3ex]
\inferrule
 {l \not\in \dom(h)}
 {\left(\HNew{x}{\var{cn}},\stb{\rho,h,A}\right) \leadsto \str{\rho[x\mapsto l],h[l \mapsto (\var{cn},o_{\vnull})],A\cup\{l\}}}
\\[1ex]
\inferrule
{~}
{\left(\HReturn{x},\stb{\rho,h,A}\right) \leadsto \str{\rho[\var{ret}\mapsto\rho(x)],h,A}}
\\[3ex]
\inferrule
{\begin{array}{c}
  h(\rho(y))=(\var{cn}_y,\underscore) \quad \lookup(\var{cn}_y,m) = \left(\Copy(X')~\HAssign{m(a)}c\right) 
  \quad \var{cn}_y \preceq \var{cn} \cr
  (c,\stb{\rho_{\vnull}[a\mapsto\rho(y)],h,\emptyset}) \leadsto \str{\rho',h',A'} \cr
\end{array}}
 {\left(\HCall{x}{\var{cn}:X}{y},\stb{\rho,h,A}\right) \leadsto \str{\rho[x\mapsto \rho'(\ret)],h',A\cup A'}}
\\[3ex]
\inferrule
{
  \begin{array}[c]{c}
\dom(h)\subseteq\dom(h')\quad \forall l\in\dom(h)\setminus\Reach{h}{\rho(y)},~ h(l)=h'(l) \cr
\forall l\in\dom(h)\setminus\Reach{h}{\rho(y)},~ 
\forall l'\in\dom(h'),  l\in\Reach{h'}{l'} \Rightarrow l'\in\dom(h)\setminus\Reach{h}{\rho(y)} \cr
v \in \{\vnull\}+\Reach{h}{\rho(y)}\cup(\dom(h')\setminus\dom(h)) 
\end{array}
}
 {\left(\UnkownCall{x}{y},\stb{\rho,h,A}\right) \leadsto \str{\rho[x\mapsto v],h',A \backslash \Reachp{h}{\rho(y)}}}\\[3ex]
\inferrule
{{\left(c_1,\stb{\rho,h,A}\right) \leadsto \str{\rho_1,h_1,A_1}}\quad
 {\left(c_2,\stb{\rho_1,h_1,A_1}\right) \leadsto \str{\rho_2,h_2,A_2}}}
{\left(c_1; c_2,\stb{\rho,h,A}\right) \leadsto \str{\rho_2,h_2,A_2}}\\[3ex]
\inferrule
{\left(c_1,\stb{\rho,h,A}\right) \leadsto \str{\rho_1,h_1,A_1}}
{\left(\HIf{c_1}{c_2},\stb{\rho,h,A}\right) \leadsto \str{\rho_1,h_1,A_1}}\qquad
\inferrule
{\left(c_2,\stb{\rho,h,A}\right) \leadsto \str{\rho_2,h_2,A_2}}
{\left(\HIf{c_1}{c_2},\stb{\rho,h,A}\right) \leadsto \str{\rho_2,h_2,A_2}}\\[3ex]
\inferrule{~~}
{\left(\HWhile{c},\stb{\rho,h,A}\right) \leadsto \str{\rho,h,A}}\qquad
\inferrule
{\left(c; \HWhile{c},\stb{\rho,h,A}\right) \leadsto \str{\rho',h',A'}}
{\left(\HWhile{c},\stb{\rho,h,A}\right) \leadsto \str{\rho',h',A'}}
\end{array}
$$}\vspace{3 pt}
\begin{minipage}[t]{.9965\linewidth}\footnotesize
\textbf{Notations:} We write $h(l,f)$ for the value $o(f)$ such that
$l\in\dom(h)$ and $h(l)=o$. We write $h[(l,f)\mapsto v]$ for the heap $h'$
that is equal to $h$ except that the $f$ field of the object at location $l$
now has value $v$.  Similarly, $\rho[x\mapsto v]$ is the environment $\rho$
modified so that $x$ now maps to $v$.  The object $o_\vnull$ is the object
satisfying $o_\vnull(f)=\vnull$ for all field $f$, and $\rho_\vnull$ is the
environment such that $\rho_\vnull(x)=\vnull$ for all variables $x$.  We
consider methods with only one parameter and name it $p$.  $\lookup$
designates the dynamic lookup procedure that, given a class name $\var{cn}$
and a method name $m$, find the first implementation of $m$ in the class
hierarchy starting from the class of name $\var{cn}$ and scanning the
hierarchy bottom-up. It returns the corresponding method declaration.
\var{ret} is a specific local variable name that is used to store the result
of each method.  $\Reach{h}{l}$ (resp. $\Reachp{h}{l}$) denotes the set of
values that are reachable from any sequence (resp. any non-empty sequence) of
fields in $h$.
\end{minipage}
\caption{Semantic Rules.}\label{fig:semrules}
\end{figure*}

The operational semantics of the language is defined (Fig.~\ref{fig:semrules})
by the evaluation relation $\leadsto$ between configurations $\comm \times
\State$ and resulting states $\State$. The set of locally allocated locations
is updated by both the $\HNew{x}{\var{cn}}$ and the $\HCall{x}{\var{cn}:X}{y}$
statements.  The execution of an unknown method call $\UnkownCall{x}{y}$
results in a new heap $h'$ that keeps all the previous objects that were not
reachable from $\rho(y)$. It assigns the variable $x$ a reference that was
either reachable from $\rho(y)$ in $h$ or that has been allocated during this
call and hence not present in $h$.

\subsection{Policies and Inheritance}

We impose restrictions on the way that inheritance can interact with
copy policies. A method being re-defined in a sub-class can impose
further constraints on how fields of the objects returned as
result should be copied. A field already annotated \emph{deep} with policy $X$ must have
the same annotation in the policy governing the re-defined
method but a field annotated as \emph{shallow} can be annotated \emph{deep}
for a re-defined method. 

\begin{defi}[Overriding Copy Policies]
A program $p$ is well-formed with respect to overriding copy policies if and
only if for any method declaration
$\Copy(X')~\HAssign{m(x)}{}\ldots$ 
that overrides (\emph{i.e.} is declared with this signature in a subclass
 of a class $\mathit{cl}$) another method declaration 
$\Copy(X)~\HAssign{m(x)}{}\ldots$ declared in $\mathit{cl}$,
we have
\[\Pi_p(X) \subseteq \Pi_p(X').\]
\end{defi}

\noindent
Intuitively, this definition imposes that the overriding copy policy is stronger than the
policy that it overrides. 
Lemma~\ref{lem:mono} below states this formally. 

\begin{exa} 
The \ttt{java.lang.Object} class provides a \ttt{clone()} method of policy
$\{\}$ (because its native \ttt{clone()} method is \emph{shallow} on all
fields).  A class $\mathtt{A}$ declaring two fields $\mathtt{f}$ and
$\mathtt{g}$ can hence override the \ttt{clone()} method and give it a
policy $\{(X,\mathtt{g})\}$.  If a class $\mathtt{B}$ extends $\mathtt{A}$ and
overrides \ttt{clone()}, it must assign it a policy of the form
$\{(X,\mathtt{g});~ \ldots ~\}$ and could declare the field $\mathtt{f}$ as
\emph{deep}. In our implementation, we let the programmer leave the policy
part that concerns fields declared in superclasses implicit, as it is
systematically inherited.  \end{exa}

\subsection{Semantics of Copy Policies}
\label{sec:policysemantics}

The informal semantics of the copy policy annotation of a method is:
\begin{quote}
A copy method satisfies a copy policy $X$ if and only if no memory cell that is
reachable from the result of this method following only fields with
\emph{deep} annotations in $X$, is reachable from another local variable of the
caller.
\end{quote}
We formalize this by giving, in Fig.~\ref{fig:annot:semantics}, a
semantics to copy policies based on access paths. An access path
consists of a variable $x$ followed by a sequence of field names $f_i$
separated by a dot. An access path $\pi$ can be evaluated to a value
$v$ in a context $\st{\rho,h}$ with a judgement
$\evalexpr{\rho}{h}{\pi}{v}$. Each path $\pi$ has a root variable
$\proot{\pi}\in\Var$. A judgement $\vdash \pi:\tau$ holds when a path
$\pi$ follows only deep fields in the policy $\tau$.  The rule
defining the semantics of copy policies can be paraphrased as follows:
For any path
$\pi$ starting in $x$ and leading to location $l$ only following deep
fields in policy $\tau$, there cannot be another path leading to the
same location $l$ which does not start in $x$.

\begin{figure}
\scriptsize\centering
\begin{minipage}[t]{.9965\linewidth}
\bf Access path syntax
\end{minipage}
$\begin{array}{rrrrl}
\pi & \in &\AccessPath  & \dasig & x \mid \pi.f
\end{array}
$
\begin{minipage}[t]{.9965\linewidth}
\bf Access path evaluation
\end{minipage}
$
\inferrule{}{\evalexpr{\rho}{h}{x}{\rho(x)}} 
\hspace*{5ex}\inferrule
  {\evalexpr{\rho}{h}{\pi}{l} \quad h(l) = o}
  {\evalexpr{\rho}{h}{\pi.f}{o(f)}} 
$
\begin{minipage}[t]{.9965\linewidth}
\bf Access path root
\end{minipage}
$\proot{x}=x \hspace*{5ex} \proot{\pi.f} = \proot{\pi}$
\begin{minipage}[t]{.9965\linewidth}
{\bf Access path satisfying a policy}\\
We suppose given $\Pi_p:\Polid\to\Pol$ the set of copy policies of the
considered program $p$. 
\end{minipage}
$\inferrule{~}{\vdash x:\tau}\hspace*{5ex}
\inferrule
{(X_1~f_1)\in\tau, (X_2~f_2)\in\Pi_p(X_1), \cdots, (X_n~f_n)\in\Pi_p(X_{n-1})
}
{\vdash x.f_1.\ldots.f_n : \tau}
$

\begin{minipage}[t]{.9965\linewidth}
{\bf Policy semantics}
\end{minipage}
$
\inferrule
{\left.\begin{array}{rl}
\forall \pi, \pi'\in\AccessPath, \forall l,l'\in\Loc, &x=\proot{\pi},\quad \proot{\pi'}\not=x,\cr
  &\evalexpr{\rho}{h}{\pi}{l}~,\quad \evalexpr{\rho}{h}{\pi'}{l'}, \cr
  &\vdash \pi : \tau 
\end{array}\right\}
\text{implies}~ l\not= l'
}
{\rho,h,x \models \tau}
$

\caption{Copy Policy Semantics}\label{fig:annot:semantics}
\end{figure}

\begin{defi}[Secure Copy Method]
\label{def:secure-copy-method}
A method $m$ is said \emph{secure} wrt. a copy signature $\Copy(X)\{\tau\}$ if
and only if for all heaps $h_1,h_2\in\Heap$, local environments
$\rho_1,\rho_2\in\Env$, locally allocated locations $A_1,A_2\in\Power(\Loc) $, and variables $x,y\in\Var$,
$$
(\HCall{x}{\var{cn}:X}{y}, \st{\rho_1,h_1,A_1}) \leadsto \st{\rho_2,h_2,A_2}
  ~~\text{implies}~~\rho_2,h_2,x \models \tau$$
\end{defi}
\noindent Note that because of virtual dispatch, the method executed
by such a call may not be the method found in $\var{cn}$ but an
overridden version of it.  The security policy requires that all
overriding implementations still satisfy the policy $\tau$.

\begin{lem}[Monotonicity of Copy Policies wrt. Overriding]\label{lem:mono} 
$$
\tau_1\subseteq\tau_2 ~\text{implies}~ 
\forall h,\rho,x,~~ \rho,h,x \models \tau_2 ~\Rightarrow \rho,h,x \models \tau_1
$$
\end{lem}
\begin{proof}

[See Coq proof \texttt{Overriding.copy\_policy\_monotony}~\cite{clone-webpage}]

Under these hypotheses, for all access paths $\pi$, $\vdash\pi:\tau_1$ implies
$\vdash\pi:\tau_2$. Thus the result holds by definition of $\models$.
\end{proof}

Thanks to this lemma, it is sufficient to prove that each method is secure wrt.~its own copy 
signature to ensure that all potential overridings will be also secure
wrt.~that copy signature. 

\subsection{Limitations of Copy Policies}
\label{sec:pol:limitations}
The enforcement of our copy policies will ensure that certain sharing
constraints are satisfied between fields of an object and its
clone. However, in the current formalism we restrict the policy to
talk about fields that are actually present in a class. The policy
does not ensure properties about fields that are added in
sub-classes. This means that an attacker could copy \emph{e.g.}, a
list by using a new field to build the list, as in the following
example. 

\begin{lstlisting}
public class EvilList<V> extends List<V>
{
    @Shallow public List<V> evilNext;

    public EvilList(V val, List<V> next) {
	super(val,null);
	this.evilNext = next; }

    public List<V> clone() {
        return new EvilList(value,evilNext); }

   // redefinition of all other methods to use the evilNext field
   // instead of next
}  
\end{lstlisting}

\noindent
The enforcement mechanism described in this article will determine
that the \ttt{clone()} method of class \ttt{EvilList} respects the copy
policy declared for the \ttt{List} class in
Section~\ref{sec:intro:pol} because this policy only speaks about the
\ttt{next} field which is set to \ttt{null}. It will fail to discover
that the class \ttt{EvilList} creates a shallow copy of lists through
the \ttt{evilNext} field. In order to prevent this attack, the policy
language must be extended, \emph{e.g.}, by adding a facility for
specifying that all fields \emph{except} certain, specifically named
fields must be copied deeply. The enforcement of such policies will likely
be able to reuse the analysis technique described below.

\section{Type and Effect System}
\label{sec:type-system}



The annotations defined in the previous section are convenient for
expressing a copy policy but are not sufficiently expressive for
reasoning about the data structures being copied.  The static
enforcement of a copy policy hence relies on a translation of
policies into a graph-based structure (that we shall call types)
describing parts of the environment of local variables and the heap
manipulated by a program. In particular, the types can express useful
alias information between variables and heap cells.  In this section,
we define the set of types, an approximation (sub-typing) relation
$\sqsubseteq$ on types, and an inference system for assigning types to
each statement and to the final result of a method.

The set of types is defined using the following symbols:
\begin{align*}
n           &\in\Node
  &t        &\in\BaseType = \Node + \{ \bot,  \topout, \top \} \\
\Gamma      &\in \Var \to \BaseType 
  &\Delta   &\in \LSG = \Node \ptofin \Field \to \BaseType \\
\Theta      &\in \Power(\Node) 
  &T        &\in  \Type = (\Var \to \BaseType)\times\LSG\times\Power(\Node) 
\end{align*}
We assume given a set $\Node$ of nodes. A value can be given a \emph{base
type} $t$ in $\Node +  \{ \bot,  \topout, \top \}$.  A node $n$ means the
value has been locally allocated and is not shared. 
The symbol $\bot$ means that the value is equal to the null reference
$\vnull$.  The symbol $\topout$ means that the value contains a location that
cannot reach a locally allocated object.  The symbol $\top$ is the specific
``no-information'' base type. 
As is standard in analysis of memory structures, we distinguish
between nodes that represent exactly one memory cell and nodes that
may represent several cells. If a node representing only one cell
has an edge to another node, then this edge can be forgotten and \emph{replaced} when we
assign a new value to the node---this is called a \emph{strong}
update. If the node represents several cells, then the assignment may
not concern all these cells and edges cannot be forgotten. We can only
\emph{add} extra out-going edges to the node---this is termed a
\emph{weak} update. In the graphical representations of types, we use
singly-circled nodes to designate "weak" nodes and doubly-circled
nodes to represent "strong" nodes. 

A type is a triplet
$T=(\Gamma,\Delta,\Theta)\in\Type$ where
\begin{description}
  \item[$\Gamma$] is a typing environment that maps (local) variables to base
  types.
  \item[$\Delta$] is a graph whose nodes are elements of $\Node$. The edges of
  the graphs are labeled with field names.  The successors of a node is a base
  type. Edges over-approximate the concrete points-to relation.
  \item[$\Theta$] is a set of nodes that represents necessarily only one
  concrete cell each. Nodes in $\Theta$ are eligible to strong update while
  others (weak nodes) can only be weakly updated.
\end{description}

\begin{TR}
\begin{exa}
\label{ex:type:graph}
The default \ttt{List} policy of Sec.~\ref{sec:intro:pol} translates into the
type
\begin{align*}
  \Gamma &= [\mathtt{res}\mapsto n_1, \mathtt{this}\mapsto \topout] \\
  \Delta &= 
  [ (n_1,\mathtt{next})\mapsto n_2,(n_2,\mathtt{next})\mapsto n_2,(n_1,\mathtt{value})\mapsto \top,(n_2,\mathtt{value})\mapsto \top ] \\
  \Theta &=  \{n_1\}.
\end{align*}
As mentioned in Sec~\ref{sec:intro:typ}, this type enjoys a graphic representation
corresponding to the right-hand side of Fig.~\ref{fig:listab}.
\end{exa}
\end{TR}


In order to link types to the heap structures they represent, we
will need to state reachability predicates in the abstract domain. Therefore,
the path evaluation relation is extended to types using the following
inference rules:
\begin{small}
\begin{equation*}
\inferrule{~}
          {\tevalexpr{\Gamma}{\Delta}{x}{\Gamma(x)}}
\quad
\inferrule{\tevalexpr{\Gamma}{\Delta}{\pi}{n}}
          {\tevalexpr{\Gamma}{\Delta}{\pi.f}{\Delta[n,f]}}
\quad
\inferrule{\tevalexpr{\Gamma}{\Delta}{\pi}{\top}}
          {\tevalexpr{\Gamma}{\Delta}{\pi.f}{\top}}
\quad
\inferrule{\tevalexpr{\Gamma}{\Delta}{\pi}{\topout}}
          {\tevalexpr{\Gamma}{\Delta}{\pi.f}{\topout}}
\end{equation*}
\end{small}
Notice both $\topout$ and $\top$ are considered as sink nodes for path
evaluation purposes
\footnote{The sink nodes status of $\top$ (resp.~$\topout$) can be
  understood as a way to state the following invariant enforced by our type
  system: when a cell points to an unspecified (resp.~foreign) part of the
  heap, all successors of this cell are also unspecified (resp.~foreign).}.

\subsection{From Annotation to Type}

The set of all copy policies $\Pi_p\subseteq\CopySignature$ can be
translated into a graph  $\Delta_p$ as described hereafter.  We assume a
naming process that associates to each policy name $X\in\Polid$ of a program a
unique node $n'_X\in\Node$.
\begin{equation*}
\Delta_p = \displaystyle\bigcup_{X:\{ (X_1,f_1);\ldots ;(X_k,f_k)\} \in \Pi_p} 
  \left[ (n'_X,f_1) \mapsto n'_{X_1}, \cdots, (n'_X,f_k) \mapsto n'_{X_k} \right]
\end{equation*}

Given this graph, a policy $\tau=\{(X_1,f_1);\ldots ;(X_k,f_k)\}$ that is declared
in a class $\var{cl}$ is translated into a triplet:
$$
\Phi(\tau)=\left(n_\tau,
                 \Delta_p\cup\left[ (n_\tau,f_1) \mapsto n'_{X_1}, \cdots, (n_\tau,f_k) \mapsto n'_{X_k} \right],
                 \{n_\tau\}\right)
$$
Note that we \emph{unfold} the possibly cyclic graph $\Delta_p$ with an extra
node $n_\tau$ in order to be able to catch an alias information between this
node and the result of a method, and hence declare $n_\tau$ as strong. Take
for instance the type in Fig.~\ref{fig:listab}: were it not for this
unfolding step, the type would have consisted only in a weak node and a $\top$
node, with the variable \ttt{res} mapping directly to the former.  Note also
that it is not necessary to keep (and even to build) the full graph $\Delta_p$
in $\Phi(\tau)$ but only the part that is reachable from $n_\tau$.


\subsection{Type Interpretation}

The semantic interpretation of types is given in
Fig.~\ref{fig:type:interpret}, in the form of a relation 
$${\InterpM{\rho}{h}{A}{\Gamma}{\Delta}{\Theta}}$$ that states when a local
allocation history $A$, a heap $h$ and an environment $\rho$ are coherent with
a type $(\Gamma,\Delta,\Theta)$. 
The interpretation judgement amounts to checking that 
\begin{inparaenum}[(i)] 
 \item for every path $\pi$ that leads to a value $v$ in the concrete
   memory and to a base type $t$ in the 
   graph, $t$ is a correct description of $v$, as formalized by the auxiliary type interpretation
   ${\Interpret{\rho}{h}{A}{\Gamma}{\Delta}{v}{t}}$;
  \item every strong node in $\Theta$ represents a uniquely reachable value in
  the concrete memory.
\end{inparaenum}
The auxiliary judgement ${\Interpret{\rho}{h}{A}{\Gamma}{\Delta}{v}{t}}$ is
defined by case on $t$. The null value is represented by any type.  The
symbol $\top$ represents any value and $\topout$ those values that do not
allow to reach a locally allocated location.  A node $n$ represents a locally
allocated memory location $l$ such that every concrete path $\pi$ that leads
to $l$ in $\st{\rho,h}$ leads to node $n$ in $(\Gamma,\Delta)$.

\begin{figure}
{\centering
\begin{minipage}[t]{.9965\linewidth}
\bf Auxiliary type interpretation
\end{minipage}
\scriptsize$$
\begin{array}{c}
\inferrule{~~}
{\Interpret{\rho}{h}{A}{\Gamma}{\Delta}{\vnull}{t}}
~~~~
\inferrule{~~}
{\Interpret{\rho}{h}{A}{\Gamma}{\Delta}{v}{\top}}
~~~~
\inferrule{
\Reach{h}{l}\cap A = \emptyset
}
{\Interpret{\rho}{h}{A}{\Gamma}{\Delta}{l}{\topout}}
\\[5ex]
\inferrule{
l \in A \qquad
n\in\dom(\Delta) \qquad
\forall \pi,~ \evalexpr{\rho}{h}{\pi}{l} ~\Rightarrow
\evalexpr{\Gamma}{\Delta}{\pi}{n}  }
{\Interpret{\rho}{h}{A}{\Gamma}{\Delta}{l}{n}}
\end{array}
$$}
\begin{minipage}[t]{.9965\linewidth}
\bf Main type interpretation
\end{minipage}
{\scriptsize$$\inferrule{
\begin{array}[t]{l}
\forall \pi, \forall t, \forall v, \cr
~~~\left.\begin{array}[c]{l}
\tevalexpr{\Gamma}{\Delta}{\pi}{t}  \cr
\evalexpr{\rho}{h}{\pi}{v}  
\end{array}\right\rbrace \Rightarrow
{\Interpret{\rho}{h}{A}{\Gamma}{\Delta}{v}{t}}
\end{array}
\\
\begin{array}[t]{l}
\forall n\in\Theta,~\forall \pi, \forall \pi', \forall l,\forall l',\cr
~~~ 
\left.\begin{array}[c]{l}
\tevalexpr{\Gamma}{\Delta}{\pi}{n}  ~\land~ 
\tevalexpr{\Gamma}{\Delta}{\pi'}{n}  \cr
\evalexpr{\rho}{h}{\pi}{l}  ~\land ~
\evalexpr{\rho}{h}{\pi'}{l'}
\end{array}\right\rbrace
\Rightarrow l = l'
\end{array}
}
{\InterpM{\rho}{h}{A}{\Gamma}{\Delta}{\Theta}}
$$}
\caption{Type Interpretation}\label{fig:type:interpret}
\end{figure}

We now establish a semantic link between policy semantics and type
interpretation. We show that if the final state of a copy method can be given a type 
of the form $\Phi(\tau)$ then this is a secure method wrt. the policy $\tau$.

\begin{thm}
\label{th:policy-type}
Let $\Phi(\tau)=(n_\tau,\Delta_\tau,\Theta_\tau)$, $\rho\in\Env,
A\in\Power(\Loc)$, and $x\in\Var$. Assume that, for all $y\in\Var$ such that $y$ is
distinct from $x$, $A$ is not reachable from $\rho(y)$ in a
given heap $h$, \emph{i.e.} $ \Reach{h}{\rho(y)}\cap A = \emptyset$.  If there
exists a state of the form $\st{\rho',h,A}$, a return variable $\res$ and a
local variable type $\Gamma'$ such that $\rho'(\res)=\rho(x)$,
$\Gamma'(\res)=n_\tau$ and
$\InterpM{\rho'}{h}{A}{\Gamma'}{\Delta_\tau}{\Theta_\tau}$, then $\rho,h,x
\models \tau$ holds.
\end{thm}
\begin{proof}

[See Coq proof \texttt{InterpAnnot.sound_annotation_to_type}~\cite{clone-webpage}]

We consider two paths $\pi'$ and $x.\pi$ such that
$\proot{\pi'}\not=x$, $\evalexpr{\rho}{h}{\pi'}{l}$, $\vdash x.\pi:\tau$,
$\evalexpr{\rho}{h}{x.\pi}{l}$ and look for a contradiction. 
Since $\vdash x.\pi:\tau$ and $\Gamma'(\res)=n_\tau$, there exists a
node $n\in\Delta_\tau$ such that 
$\tevalexpr{\Gamma'}{\Delta_\tau}{ \res.\pi}{n}$. Furthermore
$\evalexpr{\rho'}{h}{\res.\pi}{l}$ so we can deduce that $l\in A$.  Thus we
obtain a contradiction with $\evalexpr{\rho}{h}{\pi'}{l}$ because any
path that starts from a variable other than $x$ cannot reach the
elements in $A$.
\end{proof}

\subsection{Sub-typing}

\begin{figure}
\centering\scriptsize
\begin{minipage}[t]{.9965\linewidth}
\bf Value sub-typing judgment
\end{minipage}\\
$
\inferrule{t\in \BaseType}{\bot \leq_\sigma t} \quad
\inferrule{t\in \BaseType \backslash \Node}{t \leq_\sigma \top} \quad
\inferrule{~~}{\topout \leq_\sigma \topout} \quad
\inferrule{n\in\Node}{n \leq_\sigma \sigma(n)} 
$
\begin{minipage}[t]{.9965\linewidth}
\bf Main sub-typing judgment
\end{minipage}\\[-0.9cm] ~
\begin{prooftree}
  \AxiomC{
    \begin{minipage}[t]{0.75\linewidth}
    \renewcommand{\theequation}{$\text{ST}_\arabic{equation}$}
    \begin{align}
      &\sigma\in\dom (\Delta_1) \rightarrow \dom (\Delta_2) + \{\top\}
        \label{eq:st1}\\
      \begin{split}
        &\forall t_1\in\BaseType, \forall \pi\in\AccessPath, \tevalexpr{\Gamma_1}{\Delta_1}{\pi}{t_1} 
         \Rightarrow \exists t_2\in \BaseType, t_1 \leq_\sigma t_2 \wedge \tevalexpr{\Gamma_2}{\Delta_2}{\pi}{t_2} 
      \end{split}
        \label{eq:st2} \\
      &\forall n_2\in\Theta_2, ~ \exists n_1\in\Theta_1,~ \sigma^{-1}(n_2) = \{n_1\} \label{eq:st3}
    \end{align}
    \end{minipage}}
  \UnaryInfC{$(\Gamma_1,\Delta_1,\Theta_1)\sqsubseteq (\Gamma_2,\Delta_2,\Theta_2)$}
\end{prooftree}
\caption{Sub-typing}\label{fig:subtype:rules}
\end{figure}

To manage control flow merge points we rely on a sub-typing relation
$\sqsubseteq$ described in Fig.~\ref{fig:subtype:rules}. 
A sub-type relation ${(\Gamma_1,\Delta_1,\Theta_1)\sqsubseteq
(\Gamma_2,\Delta_2,\Theta_2)}$ holds if and only if \eqref{eq:st1} there
exists a fusion function $\sigma$ from $\dom(\Delta_1)$ to $\dom(\Delta_2)+\{\top\}$.
$\sigma$ is a mapping that merges nodes and edges in $\Delta_1$ such that 
\eqref{eq:st2} every element $t_1$ of $\Delta_1$ accessible from a path $\pi$
is mapped to an element $t_2$ of $\Delta_2$ accessible from the same path,
such that $t_1 \leq_\sigma t_2$. In particular, this means that all successors of
$t_1$ are mapped to successors of $t_2$. Incidentally, because $\top$ acts as
a sink on paths, if $t_1$ is mapped to $\top$, then all its successors are
mapped to $\top$ too.
Finally, when a strong node in $\Delta_1$ maps to a strong node in
$\Delta_2$, this image node cannot be the image of any other node in
$\Delta_1$---in other terms, $\sigma$ is injective on strong nodes \eqref{eq:st3}.

Intuitively, it is possible to go up in the type partial order either by
merging, or by forgetting nodes in the initial graph. 
The following example shows three ordered types and their corresponding fusion functions.
On the left, we forget the node pointed to by $\mathtt{y}$ and hence forget
all of its successors (see~\eqref{eq:st2}).
On the right we fusion two strong nodes to obtain a weak node. 
\begin{center}
\includegraphics[width=.8\textwidth]{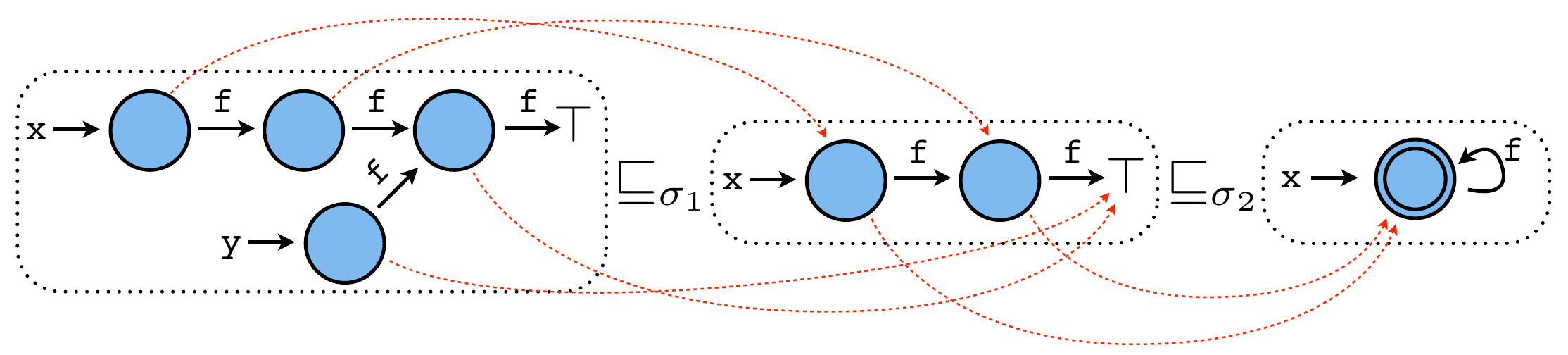}  
\end{center}

\noindent
The logical soundness of this sub-typing relation is formally proved with
two intermediate lemmas. 
The first one states that paths are preserved between subtypes, and that
they evaluate into basetypes that are related by the subtyping function.
\begin{lem}[Pathing in subtypes]
\label{lem:subtype}
Assume $\InterpMMM{\rho}{h}{A}{T_1}$, and let $\sigma$ be the fusion map
defined by the assertion $T_1\sqsubseteq T_2$. For any $\pi,r$ such that
$\evalexpr{\rho}{h}{\pi}{r}$, for any $t_2$ such that
$\tevalexpr{\Gamma_2}{\Delta_2}{\pi}{t_2}$:
\begin{align*}
\exists t_1 \leq_\sigma t_2,\quad
\tevalexpr{\Gamma_1}{\Delta_1}{\pi}{t_1}. 
\end{align*}
\end{lem}

\proof

[See Coq proof \texttt{Misc.Access_Path_Eval_subtyp}~\cite{clone-webpage}]
The proof follows directly from the definition of type interpretation and subtyping.\qed

The second lemma gives a local view on logical soundness of subtyping.
\begin{lem}[Local logical soundness of subtyping]
\label{lem:subtype2}
Assume $(\Gamma_1,\Delta_1,\Theta_1)\sqsubseteq
(\Gamma_2,\Delta_2,\Theta_2)$, and let $v$ be a value
and $t_1,t_2$ some types.
$$
{\InterpretS{\rho}{h}{A}{(\Gamma_1,\Delta_1)}{v}{t_1}}
~~
\text{and}
~~t_1 \sqsubseteq_\sigma t_2~~
\text{implies}
~~
{\InterpretS{\rho}{h}{A}{(\Gamma_2,\Delta_2)}{v}{t_2}}.
$$
\end{lem}
\proof

[See Coq proof \texttt{Misc.Interp_monotone}~\cite{clone-webpage}]
We make a case for each rules of $t_1 \sqsubseteq_\sigma t_2$. The only non-trivial case
is for $v=l\in\Loc$, $t_1 = n \in\Node$ and $t_2 = \sigma(n) \in \dom(\Delta_2)$.
In this case we have to prove 
$\forall \pi,~ \evalexpr{\rho}{h}{\pi}{l} ~\Rightarrow
\tevalexpr{\Gamma_2}{\Delta_2}{\pi}{\sigma(n)} $.
Given such a path $\pi$, the hypothesis 
${\InterpretS{\rho}{h}{A}{(\Gamma_1,\Delta_1)}{l}{n}}$ gives
us $\tevalexpr{\Gamma_1}{\Delta_1}{\pi}{n}$. Then subtyping hypothesis $ST_2$  
gives us a base type $t'_2$ such that 
$n \sqsubseteq_\sigma t'_2$ and $\tevalexpr{\Gamma_2}{\Delta_2}{\pi}{t'_2}$.
But necessarily $t_2=t'_2$ so we are done. \qed

The logical soundness of this sub-typing relation is then formally proved with the
following theorem.

\begin{thm} For any type $T_1,T_2\in\Type$ and $\st{\rho,h,A}\in \State$,
${T_1\sqsubseteq T_2}$
and
$\InterpMMM{\rho}{h}{A}{T_1}$
imply  
$\InterpMMM{\rho}{h}{A}{T_2}$.
\end{thm}
\begin{proof}
  \begin{CA}
  See~\cite{JensenKP10} and the companion Coq development.
  \end{CA}
  \begin{TR}[See Coq proof \texttt{InterpMonotony.Interpretation_monotone}~\cite{clone-webpage}]

We suppose $T_1$ is of the form $(\Gamma_1,\Delta_1,\Theta_1)$
and $T_2$ of the form $(\Gamma_2,\Delta_2,\Theta_2)$.
From the definition of the main type interpretation (Fig~\ref{fig:type:interpret}), we reduce
the proof to proving the following two subgoals. 

First, given a path $\pi$, a base type $t_2$ and a value $v$ such that
$\evalexpr{\Gamma_2}{\Delta_2}{\pi}{t_2}$
and
$\evalexpr{\rho}{h}{\pi}{v}$, we must prove that
${\InterpretS{\rho}{h}{A}{(\Gamma_2,\Delta_2)}{v}{t_2}}$ holds.
Since $(\Gamma_1,\Delta_1,\Theta_1)\sqsubseteq (\Gamma_2,\Delta_2,\Theta_2)$, 
there exists, by Lemma~\ref{lem:subtype}, a base type $t_1$ such that
$\tevalexpr{\Gamma_1}{\Delta_1}{\pi}{t_1}$ and 
$t_1 \sqsubseteq_\sigma t_2$. Since $\InterpMMM{\rho}{h}{A}{T_1}$ holds we
can argue that ${\InterpretS{\rho}{h}{A}{(\Gamma_1,\Delta_1)}{v}{t_1}}$ holds
too and conclude with Lemma~\ref{lem:subtype2}.

Second, given a strong node $n_2\in\Theta_2$, two paths $\pi$ and $\pi'$ and two locations
$l$ and $l'$ such that 
$\tevalexpr{\Gamma_2}{\Delta_2}{\pi}{n_2}$,
$\tevalexpr{\Gamma_2}{\Delta_2}{\pi'}{n_2}$,
$\evalexpr{\rho}{h}{\pi}{l}$ and
$\evalexpr{\rho}{h}{\pi}{l'}$, we must prove that $l=l'$.
As previously, there exists by Lemma~\ref{lem:subtype}, $t_1$ and $t'_1$ such that 
$\tevalexpr{\Gamma_1}{\Delta_1}{\pi}{t_1}$, $t_1 \sqsubseteq_\sigma n_2$,
$\tevalexpr{\Gamma_1}{\Delta_1}{\pi'}{t'_1}$ and 
$t'_1 \sqsubseteq_\sigma n_2$. But then, by $(ST_3$), there exists some strong node
$n_1$ such that $t_1 = t'_1 = n_1$ and we can obtain the desired
equality from the hypothesis 
$\InterpMMM{\rho}{h}{A}{\Gamma_1,\Delta_1,\Theta_1}$.
  \end{TR}
\end{proof}

\subsection{Type and Effect System}
\label{sec:typesystem}

The type system verifies, statically and class by class, that a program
respects the copy policy annotations relative to a declared copy policy. 
%
The core of the type system concerns the typability of commands, which is
defined through the following judgment:
$$
\Gamma,\Delta,\Theta \vdash c : \Gamma',\Delta',\Theta'.
$$
The judgment is valid if the execution of command $c$ in a state satisfying
type $(\Gamma,\Delta,\Theta)$ will result in a state satisfying
$(\Gamma',\Delta',\Theta')$ or will diverge. 

\begin{figure}
\centering\scriptsize
\begin{minipage}[t]{.9965\linewidth}
\bf Command typing rules
\end{minipage}\\
$
\begin{array}{c}
\inferrule{~~}
  {\tstb{\Gamma,\Delta,\Theta} \vdash \HAssign{x}{y} : \tstr{\Gamma[\var{x}\mapsto\Gamma(\var{y})],\Delta,\Theta}}
\quad
\inferrule{n~\text{fresh in}~\Delta}
  {\tstb{\Gamma,\Delta,\Theta} \vdash \HNew{x}{\var{cn}} : \tstr{\Gamma[\var{x}\mapsto n],\Delta[(n,\_)\mapsto\bot],\Theta\cup\{n\}}}
\\[3ex]
\inferrule{\Gamma(\var{y})=t\qquad t\in\{\topout,\top\}}
  {\tst{\Gamma,\Delta,\Theta} \vdash \HAssign{x}{y.f} : \tstr{\Gamma[\var{x}\mapsto t],\Delta,\Theta}}
\quad
\inferrule{\Gamma(\var{y})=n}
  {\tstb{\Gamma,\Delta,\Theta} \vdash \HAssign{x}{y.f} : \tstr{\Gamma[\var{x}\mapsto \Delta[n,f]],\Delta,\Theta}}
\\[3ex]
\inferrule{\Gamma(\var{x})=n\quad n\in\Theta}
  {\tstb{\Gamma,\Delta,\Theta} \vdash \HAssign{x.f}{y} : \tstr{\Gamma,\Delta[n,f \mapsto \Gamma(y)],\Theta}}
\\[3ex]
\inferrule{\Gamma(\var{x})=n\quad n\not\in\Theta
 \quad (\Gamma,\Delta[n,f \mapsto \Gamma(y)],\Theta) \sqsubseteq (\Gamma',\Delta',\Theta')
 \quad (\Gamma,\Delta,\Theta) \sqsubseteq (\Gamma',\Delta',\Theta')
}
  {\tstb{\Gamma,\Delta,\Theta} \vdash \HAssign{x.f}{y} : \tstr{\Gamma',\Delta',\Theta'}}
\\[3ex]
\inferrule{
  \begin{array}[c]{c}
\tstb{\Gamma,\Delta,\Theta} \vdash c_1 : \tstr{\Gamma_1,\Delta_1,\Theta_1} \quad 
(\Gamma_1,\Delta_1,\Theta_1)\sqsubseteq (\Gamma',\Delta',\Theta') \cr
\tstb{\Gamma,\Delta,\Theta} \vdash c_2 : \tstr{\Gamma_2,\Delta_2,\Theta_2} \quad
(\Gamma_2,\Delta_2,\Theta_2)\sqsubseteq (\Gamma',\Delta',\Theta')    
\end{array}}
  {\tstb{\Gamma,\Delta,\Theta} \vdash \HIf{c_1}{c_2} : \tstr{\Gamma',\Delta',\Theta'}}
\\[3ex]
\inferrule{\tstb{\Gamma',\Delta',\Theta'} \vdash c : \tstr{\Gamma_0,\Delta_0,\Theta_0} \quad (\Gamma,\Delta,\Theta)\sqsubseteq (\Gamma',\Delta',\Theta') \quad (\Gamma_0,\Delta_0,\Theta_0)\sqsubseteq (\Gamma',\Delta',\Theta')}
  {\tstb{\Gamma,\Delta,\Theta} \vdash \HWhile{c} : \tstr{\Gamma',\Delta',\Theta'}}
\\[3ex]
\inferrule{\tstb{\Gamma,\Delta,\Theta} \vdash c_1 : \tstr{\Gamma_1,\Delta_1,\Theta_1} \quad \tstb{\Gamma_1,\Delta_1,\Theta_1} \vdash c_2 : \tstr{\Gamma_2,\Delta_2,\Theta_2}}
  {\tstb{\Gamma,\Delta,\Theta} \vdash c_1; c_2 : \tstr{\Gamma_2,\Delta_2,\Theta_2}}
\\[4ex]
\inferrule{
\Pi_p(X) = \tau \quad \Phi(\tau) = (n_\tau,\Delta_\tau) \quad
\nodes(\Delta)\cap \nodes(\Delta_\tau) = \emptyset \quad
(\Gamma(y) = \bot) \vee (\Gamma(y)=\topout)}
  {\tstb{\Gamma,\Delta,\Theta} \vdash \HCall{x}{\var{cn}:X}{y} : \tstr{\Gamma[x\mapsto n_\tau],\Delta\cup\Delta_\tau,\Theta\cup\{n_\tau\}}}
\\[4ex]
\inferrule{
  \begin{array}[c]{c}
\Pi_p(X) = \tau \quad \Phi(\tau) = (n_\tau,\Delta_\tau) \quad
\nodes(\Delta)\cap \nodes(\Delta_\tau) = \emptyset \cr
\mathit{KillSucc}_{n}(\Gamma,\Delta,\Theta) = (\Gamma',\Delta',\Theta')\quad
\Gamma(y) = n
\end{array}}
  {\tstb{\Gamma,\Delta,\Theta} \vdash \HCall{x}{\var{cn}:X}{y} : \tstr{\Gamma'[x\mapsto n_\tau],\Delta'\cup\Delta_\tau,\Theta'\cup\{n_\tau\}}}
\\[4ex]
\inferrule{(\Gamma(y) = \bot) \vee (\Gamma(y)=\topout)}
  {\tstb{\Gamma,\Delta,\Theta} \vdash \UnkownCall{x}{y} : \tstr{\Gamma[x\mapsto \topout],\Delta,\Theta}}
\quad
\inferrule{
\mathit{KillSucc}_{n}(\Gamma,\Delta,\Theta) = (\Gamma',\Delta',\Theta')\quad
\Gamma(y) = n}
  {\tstb{\Gamma,\Delta,\Theta} \vdash \UnkownCall{x}{y} : \tstr{\Gamma'[x\mapsto \topout],\Delta',\Theta'}}
\\
\inferrule{~~}
  {\tstb{\Gamma,\Delta,\Theta} \vdash \HReturn{x} : \tstr{\Gamma[\var{ret}\mapsto\Gamma(x)],\Delta,\Theta}}
\end{array}
$\\
\begin{minipage}[t]{.9965\linewidth}
\bf Method typing rule
\end{minipage}\\
$
  \inferrule{
    \begin{array}[c]{c}
 [~\cdot\mapsto\bot][x\mapsto \topout], \emptyset, \emptyset \vdash c : \Gamma,\Delta,\Theta \cr
   \Pi_p(X) = \tau \quad \Phi(\tau) = (n_\tau,\Delta_\tau) \quad
     (\Gamma, \Delta, \Theta) \sqsubseteq (\Gamma',\Delta_\tau,\{n_\tau\})
\quad \Gamma'(\var{ret}) =  n_\tau 
   \end{array}}
  {  \vdash \Copy(X)~\HAssign{m(x)}{}c}
$\\
\begin{minipage}[t]{.9965\linewidth}
\bf Program typing rule
\end{minipage}
$
  \inferrule{\forall \var{cl}\in p, ~ \forall\md\in\var{cl},~~  \vdash  \md }
  {\vdash p}
$
\vspace{10 pt}

\begin{minipage}[t]{.9965\linewidth}\footnotesize
\textbf{Notations:} We write
$\Delta[(n,\_)\mapsto\bot]$ for the update of $\Delta$ with a new node $n$ for
which all successors are equal to $\bot$.
We write
$\mathit{KillSucc}_{n}$ for the function that removes all
nodes reachable from $n$ (with at least one step) and sets all its successors
equal to $\top$.
\end{minipage}
\caption{Type System}\label{fig:type:rules}
\end{figure}

Typing rules are given in Fig.~\ref{fig:type:rules}. We explain a selection of
rules below.  The rules for $\HIf{}{}$, $\HWhile{}$, sequential
composition and most of the assignment rules are standard for flow-sensitive
type systems. The rule for $\HNew{x}{}$ ``allocates'' a fresh node $n$ with no
edges in the graph $\Delta$ and let $\Gamma(x)$ references this node. 

There are two rules concerning the instruction $\HAssign{x.f}{y}$ for
assigning values to fields. Assume that the variable $x$ is
represented by node $n$ (\emph{i.e.}, $\Gamma(x) = n$).  In the first
case (strong update), the node is strong and we update destructively (or add) the edge
in the graph $\Delta$ from node $n$ labeled $f$ to point to the value
of $\Gamma(y)$. The previous edge (if any) is lost because 
$n\in\Theta$ ensures that all concrete cells represented by $n$ are
affected by this field update. In the second case (weak update), the node is
weak. In order to be conservative, we must merge the previous shape
with its updated version since the content of $x.f$ is updated but an
other cell mays exist and be represented by $n$ without being
affected by this field update.

As for method calls $m(\var{y})$, two cases arise depending on whether the
method $m$ is copy-annotated or not. In each case, we also reason
differently depending on the type of the argument $\var{y}$.  If a
method is associated with a copy policy $\tau$, we compute the
corresponding type $(n_\tau,\Delta_\tau)$ and type the result of
$\HCall{x}{\var{cn}:X}{y}$ starting in $(\Gamma,\Delta,\Theta)$ with
the result type consisting of the environment $\Gamma$ where $x$ now
points to $n_\tau$, the heap described by the disjoint union of
$\Delta$ and $\Delta_\tau$. In addition, the set of strong nodes is
augmented with $n_\tau$ since a copy method is guaranteed to return a
freshly allocated node. The method call may potentially modify the
memory referenced by its argument $\var{y}$, but the analysis has no
means of tracking this. Accordingly, if  $\var{y}$ is a locally allocated memory
location of type $n$, we must remove all nodes reachable from $n$, and
set all the successors of $n$ to $\top$.  The other case to consider
is when the method is
not associated with a copy policy (written $ \UnkownCall{x}{y}$).  If the parameter $\var{y}$ is null
or not locally allocated, then there is no way for the method call to
access locally allocated memory and we know that $\var{x}$ points to a
non-locally allocated object. Otherwise, $\var{y}$ is a locally allocated
memory location of type $n$, and we must kill all its successors in
the abstract heap.

Finally, the rule for method definition verifies the coherence of the
result of analysing the body of a method $m$
with its copy annotation $\Phi(\tau)$. Type checking extends trivially to
all methods of the program. 

Note the absence of a rule for typing an instruction $\HAssign{x.f}{y}$ when
$\Gamma(\var{x})=\top$ or $\topout$. In a first attempt, a sound rule would
have been
\begin{small}
\begin{equation*}
\inferrule
  {\Gamma(\var{x})=\top}
  {\Gamma,\Delta \vdash \HAssign{x.f}{y} : \Gamma,\Delta[~\cdot,f \mapsto \top]} 
\end{equation*}
\end{small}%
Because $\var{x}$ may point to any part of the local shape we must
conservatively forget all knowledge about the field
$\var{f}$. Moreover we should also warn the caller of the current
method that a field $\var{f}$ of his own local shape may have been
updated.  We choose to simply reject copy methods with such patterns. 
Such a policy is strong but has the merit to be easily
understandable to the programmer: a copy method should only modify
locally allocated objects to be typable in our type system. For
similar reasons, we reject methods that attempt to make a method call on
a reference of type $\top$ because we can not track side effect
modifications of such methods without losing the modularity of the verification mechanism.

\begin{figure}
  \centering
  \begin{minipage}{.5\linewidth}
    \lstset{basicstyle=\ttfamily\scriptsize,numbers=left,numberstyle=\scriptsize,numbersep=2pt}
    \begin{lstlisting}
class LinkedList<E> implements Cloneable {
   private @Deep Entry<E> header;
 
   private static class Entry<E> {
     @Shallow E element;
     @Deep Entry<E> next;
     @Deep Entry<E> previous;
  }

  @Copy public Object clone() {
    LinkedList<E> clone = null;
    clone = (LinkedList<E>) super.clone();
    clone.header = new Entry<E>;
    clone.header.next = clone.header;
    clone.header.previous = clone.header;
    Entry<E> e = this.header.next;
    while (e != this.header) {
      Entry<E> n = new Entry<E>;
      n.element = e.element;
      n.next = clone.header;
      n.previous = clone.header.previous;
      n.previous.next = n;
      n.next.previous = n;
      e = e.next;
    }
    return clone;
  } 
}
\end{lstlisting}  
  \end{minipage}
  \begin{minipage}{.49\linewidth}
    \includegraphics[width=\linewidth]{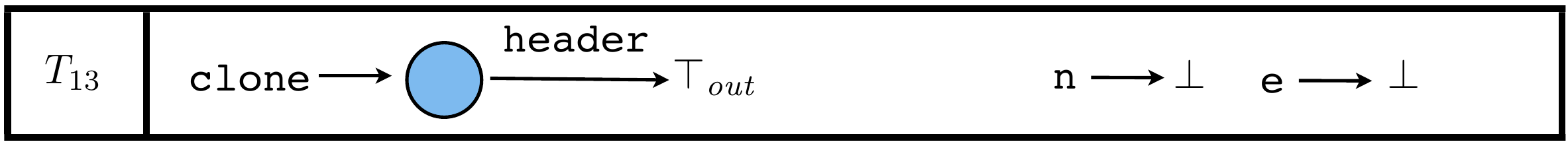}
    \includegraphics[width=\linewidth]{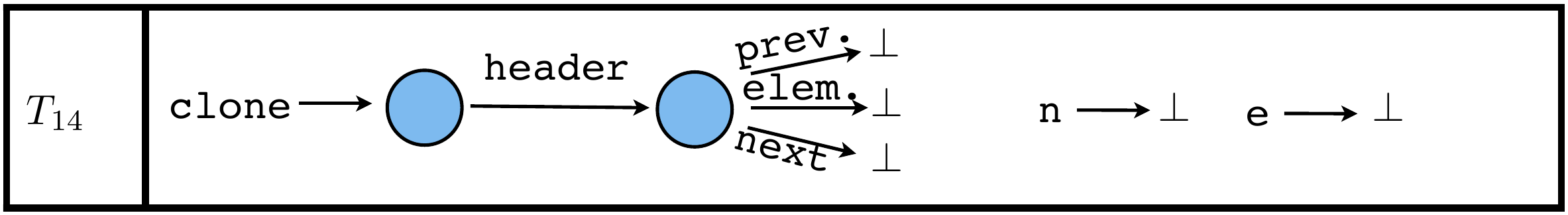}
    \includegraphics[width=\linewidth]{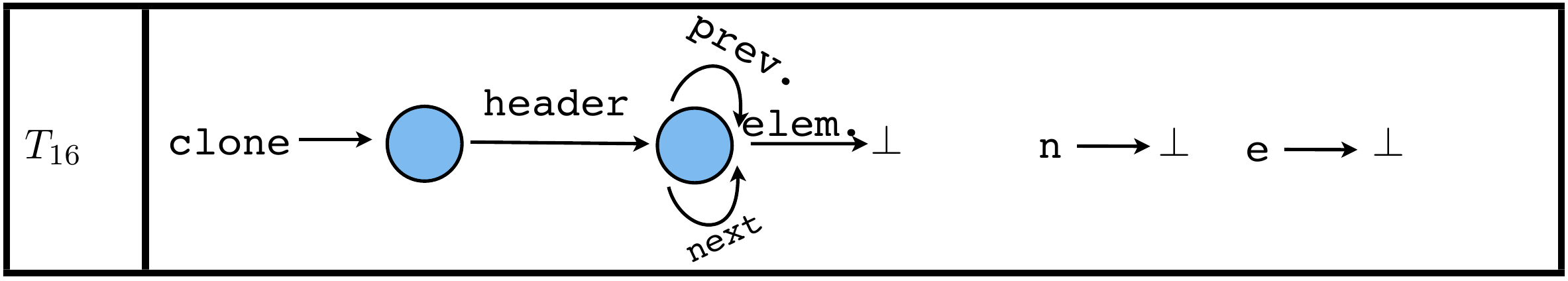}
    \includegraphics[width=\linewidth]{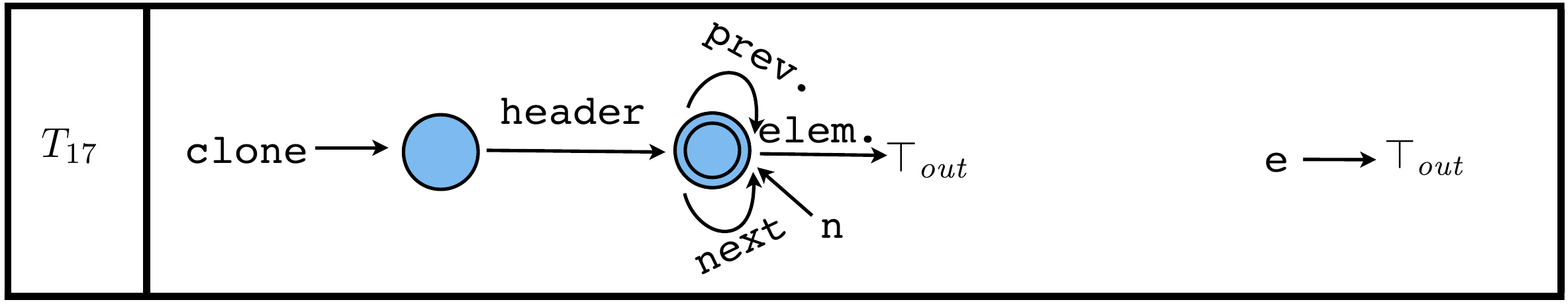}
    \includegraphics[width=\linewidth]{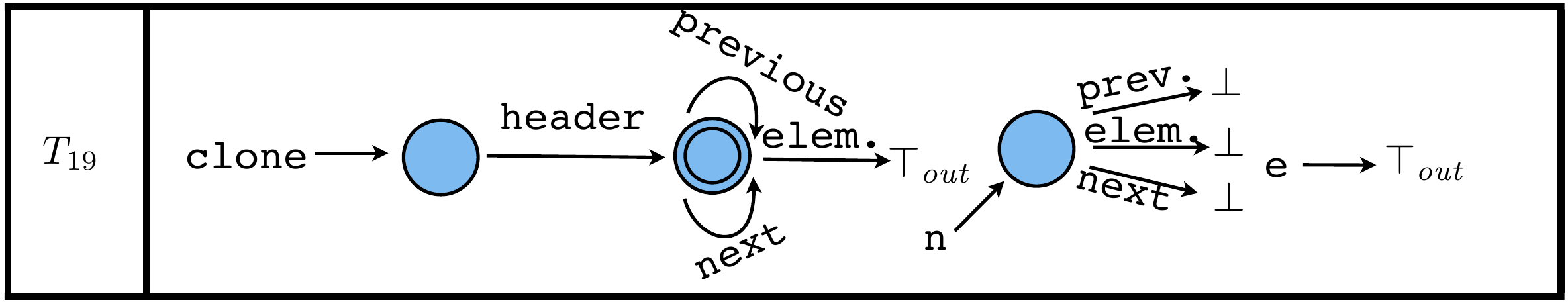}
    \includegraphics[width=\linewidth]{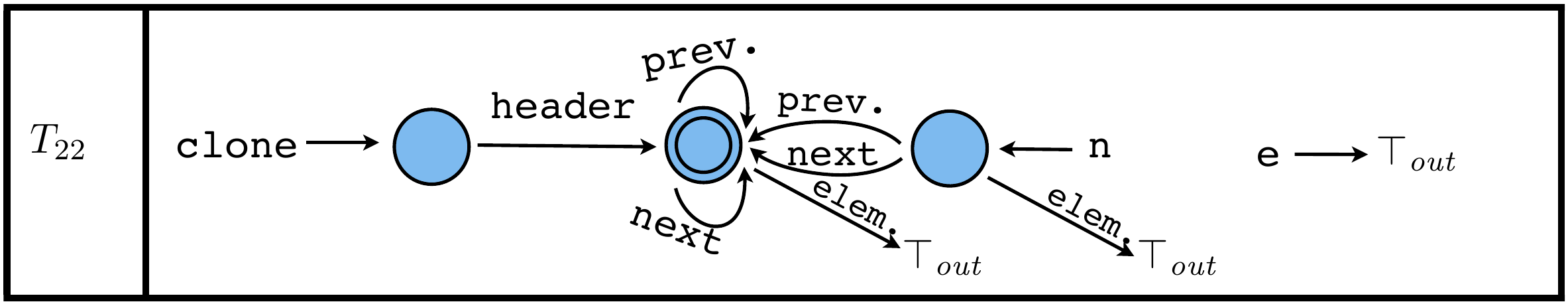}
    \includegraphics[width=\linewidth]{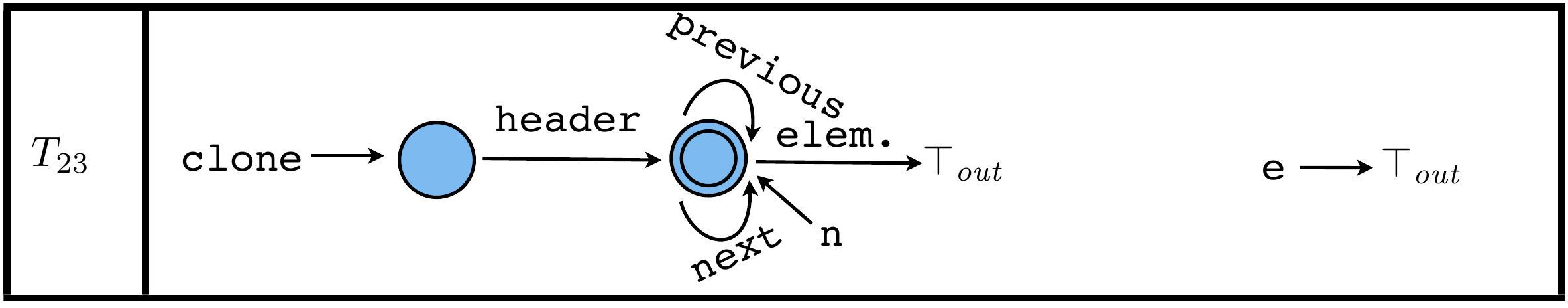}
    \includegraphics[width=\linewidth]{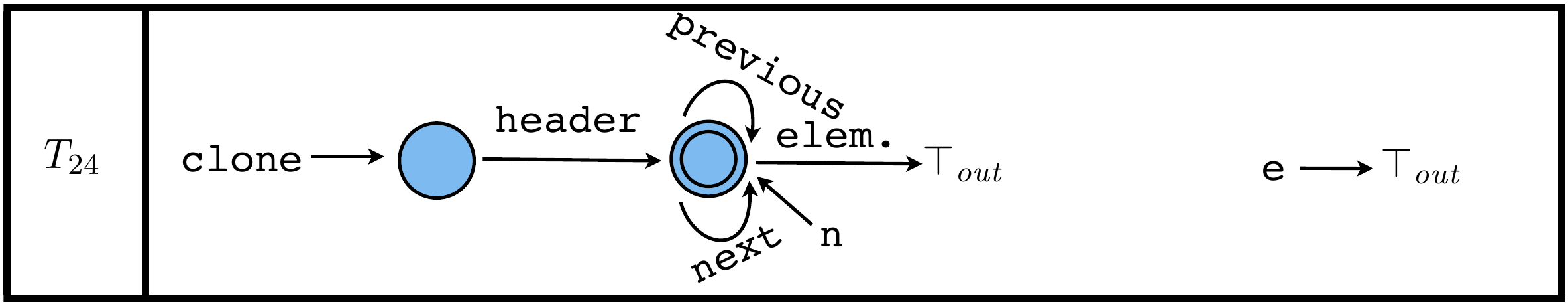}
\end{minipage}
  \caption{Intermediate Types for \ttt{java.util.LinkedList.clone()}}
  \label{fig:linkedlist}
\end{figure}

\begin{exa}[Case Study: \ttt{java.util.LinkedList}]
In this example, we demonstrate the use of the type system on a challenging
example taken from the standard Java library.  
The companion
web page provides a more detailed explanation of this example~\cite{clone-webpage}.  
The class
\ttt{java.util.LinkedList} provides an implementation of doubly-linked lists. 
A list is composed of a first cell that points through a field
\ttt{header} to a collection of doubly-linked cells.  Each cell has a
link to the previous and the next cell and also to an element of
(parameterized) type \ttt{E}.  The clone method provided in
\ttt{java.lang} library implements a ``semi-shallow'' copy where only cells of
type \ttt{E} may be shared between the source and the result of the
copy.  
In Fig.~\ref{fig:linkedlist} we present a modified version of the
original source code: we have inlined all method calls, except those
to copy methods and removed exception handling that leads to an
abnormal return from the method\footnote{Inlining is automatically
  performed by our tool and exception control flow graph is managed as
  standard control flow but omitted here for simplicity.}.  Note that
there was one method call in the original code that was virtual and
hence prevented inlining.  It has been necessary to make a private
version of this method. This makes sense because such a virtual call
actually constitutes a potentially dangerous hook in a cloning method,
as a re-defined implementation could be called when cloning a subclass
of \ttt{Linkedlist}.


In Fig.~\ref{fig:linkedlist} we provide several intermediate types
that are necessary for typing this method ($T_i$ is the type before
executing the instruction at line $i$). The call
to \ttt{super.clone} at line 12 creates a shallow copy of the header
cell of the list, which contains a reference to the original list. The
original list is thus shared, a fact which is 
represented by an edge to $\topout$ in type $T_{13}$.


The copy method then progressively constructs a deep copy of the list,
by allocating a new node (see type $T_{14}$) and setting all paths
\ttt{clone.header}, \ttt{clone.header.next} and
\ttt{clone.header.previous} to point to this node. This is reflected
in the analysis by a \emph{strong update} to the node representing
path \ttt{clone.header} to obtain the type $T_{16}$ that precisely
models the alias between the three paths \ttt{clone.header},
\ttt{clone.header.next} and \ttt{clone.header.previous} (the Java
syntax used here hides the temporary variable that is introduced to be
assigned the value of \ttt{clone.header} and then be updated).
 
This type $T_{17}$ is the loop invariant necessary for type checking
the whole loop. It is a super-type of $T_{16}$ (updated with
$e\mapsto\topout$) and of $T_{24}$ which represents the memory at the
end of the loop body.  The body of the loop allocates a new list cell
(pointed to by variable \ttt{n}) (see type $T_{19}$) and inserts it
into the doubly-linked list. The assignment in line 22 updates the
weak node pointed to by path \ttt{n.previous} and hence merges the
strong node pointed to by \ttt{n} with the weak node pointed to by
\ttt{clone.header}, representing the spine of the list. The
assignment at line 23 does not modify the type $T_{23}$.



Notice that the types used in this example show that a flow-insensitive
version of the analysis could not have found this information. A
flow-insensitive analysis would force the merge of the types at all program
points, and the call to \ttt{super.clone} return a type that is less precise
than the types needed for the analysis of the rest of the method.

\end{exa}


\subsection{Type soundness}
\label{sec:soudness}

The rest of this section is devoted to the soundness proof of the type system.
We consider the types
$T=(\Gamma,\Delta,\Theta),
T_1=(\Gamma_1,\Delta_1,\Theta_1),T_2=(\Gamma_2,\Delta_2,\Theta_2)\in\Type$, a
program $c\in\prog$, as well as the configurations $\st{\rho,h,A},
\st{\rho_1,h_1,A_1}, \st{\rho_2,h_2,A_2}\in\State$.


Assignments that modify the heap space can also modify the reachability
properties of locations. This following lemma indicates how to reconstruct a path to $l_f$ in the \emph{initial} heap from a path
$\pi'$ to a given location $l_f$ in the \emph{assigned} heap. 
\begin{lem}[Path decomposition on assigned states]
\label{lem:pathdec}
Assume given a path $\pi$, field $f$, and locations $l,l'$ such that
$\evalexpr{\rho}{h}{\pi}{l'}$. 
and assume that for any path $\pi'$ and a location $l_f$ we have
$\evalexpr{\rho}{h[l,f\mapsto l']}{\pi'}{l_f}$. Then, either
\begin{equation*}
\evalexpr{\rho}{h}{\pi'}{l_f}
\end{equation*}
or $\exists \pi_z, \pi_1, \ldots, \pi_n, \pi_f$ such that:
\begin{equation*}
\left\{
\begin{aligned}
  &\pi'=\pi_z.f.\pi_1.f.\ldots.f.\pi_n.f.\pi_f \\
  &\evalexpr{\rho}{h}{\pi_z}{l}\\
  &\evalexpr{\rho}{h}{\pi.\pi_f}{l_f}\\
  &\forall \pi_1,\ldots,\pi_n, \evalexpr{\rho}{h}{\pi.\pi_i}{l}
\end{aligned}
\right.
\end{equation*}
The second case of the conclusion of this lemma is illustrated in
Fig.~\ref{fig:pathdec}.
\end{lem}
\proof 

[See Coq proof \texttt{Misc.Access_Path_Eval_putfield_case}~\cite{clone-webpage}] 
The proof is done by induction on $\pi$. \qed

\begin{figure}%
  \centering
  \subfloat[in $\st{\rho,h}$]{\includegraphics[width=.35\textwidth]{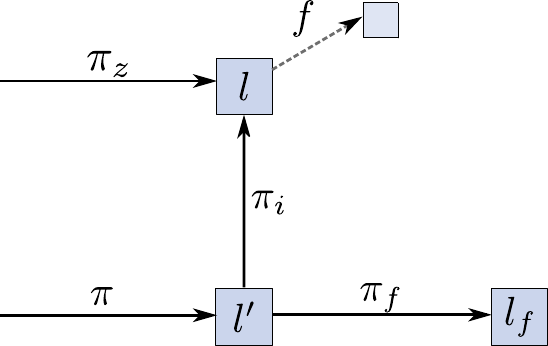}}\qquad
  \subfloat[in $\st{\rho,h[l,f\mapsto l']}$]{\includegraphics[width=.35\textwidth]{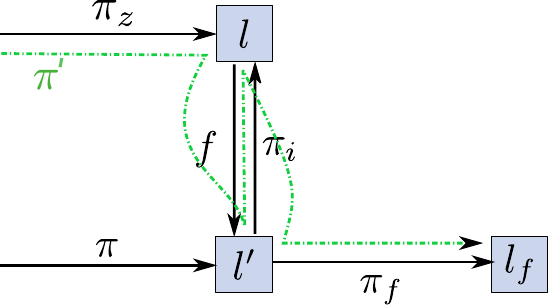}}
  \caption{An Illustration of Path Decomposition on Assigned States.}
  \label{fig:pathdec}
\end{figure}

We extend the previous lemma to paths in both the concrete heap and the graph types.
\begin{lem}[Pathing through strong field assignment]
\label{lem:strongassign}
Assume $\InterpMMM{\rho}{h}{A}{\Gamma,\Delta,\Theta}$ with $\rho(x)=l_x\in A$,
$\Gamma(x)=n_x\in \Theta$, $\rho(y)=l_y$, and $\Gamma(y)=t_y$. Additionally,
suppose that for some path $\pi$, value $v$, and type $t$:
\begin{gather*}
  \evalexpr{\rho}{h[l_x,f\mapsto l_y]}{\pi}{v} \\
  \tevalexpr{\Gamma}{\Delta[n_x,f\mapsto t_y]}{\pi}{t}.
\end{gather*}
Then at least one of the following four statements hold:
\begin{align*}
  &\left(\evalexpr{\rho}{h}{\pi}{v} \wedge \tevalexpr{\Gamma}{\Delta}{\pi}{t}\right) \tag{1}\\
  &\left(\exists\pi', \evalexpr{\rho}{h}{y.\pi'}{v} \wedge \tevalexpr{\Gamma}{\Delta}{y.\pi'}{t}\right) \tag{2}\\
  &t=\top \tag{3} \\
  &v=\vnull \tag{4}
\end{align*}
\end{lem}

\proof

[See Coq proof \texttt{Misc.strong_subst_prop}~\cite{clone-webpage}]
The non-trivial part of the lemma concerns the situation when $t \neq
\top$ and $v \neq \vnull$). In that case, 
the proof relies on Lemma~\ref{lem:pathdec}. The two parts of the disjunction
in Lemma~\ref{lem:pathdec} are used to prove one of the two first
statements. If the first
part of the disjunction holds, we can assume that $\evalexpr{\rho}{h}{\pi}{v}$. Then, since
$\InterpMMM{\rho}{h}{A}{\Gamma,\Delta,\Theta}$, we also have
$\tevalexpr{\Gamma}{\Delta}{\pi}{t}$. This implies the first main statement.
If the second part of the disjunction holds, then, by observing that
$\evalexpr{\rho}{h}{y}{l_y}$, we can derive the sub-statement $\exists \pi_f,
\evalexpr{\rho}{h}{y.\pi_f}{v}$. As previously, by assumption we also have
$\tevalexpr{\Gamma}{\Delta}{y.\pi_f}{t}$, which implies our second main
statement.
\qed

We first establish a standard subject reduction theorem and then prove type
soundness. We assume that all methods of the considered program are
well-typed.
\begin{thm}[Subject Reduction]\label{theo2}
Assume ${T_1 \vdash c: T_2}$ and 
$\InterpMM{\rho_1}{h_1}{A_1}{T_1}$.\\
If 
$(c,\st{\rho_1,h_1,A_1}) \leadsto \st{\rho_2,h_2,A_2}$
then  
$\InterpMM{\rho_2}{h_2}{A_2}{T_2}$.
\end{thm}
\proof

[See Coq proof \texttt{Soundness.subject_reduction}~\cite{clone-webpage}]
The proof proceeds by structural induction on the instruction $c$. For
each reduction rule concerning $c$ (Fig.~\ref{fig:semrules}), we prove
that the resulting state is in relation to the type $T_2$, as defined
by the main type interpretation rule in
Fig.~\ref{fig:type:interpret}. This amounts to verifying that the two
premises of the type interpretation rule are satisfied. One premise 
checks that all access paths lead to related (value,node) pairs. The
other checks that all nodes that are designated as ``strong'' in the
type interpretation indeed only represent unique locations.

We here present the most intricate part of the proof, which concerns
graph nodes different from $\top$, $\bot$, or $\topout$, and focus here on variable and field assignment. 
The entire proof has been checked using the Coq proof management system.

%
\underline{If $c\equiv \HAssign{x}{y}$}
  then $\st{\rho_2,h_2,A_2} = \st{\rho_1[x\mapsto\rho_1(y)],h_1,A_1}$.
  In $(\rho_2,h_2,\Gamma_2,\Delta_2)$ take a path $x.\pi$: since
  $\rho_2(x)=\rho_1(y)$ and $\Gamma_2(x)=\Gamma_1(y)$, $y.\pi$ is also a path
  in $(\rho_1,h_1,\Gamma_1,\Delta_1)$.  Given that
  $\InterpMMM{\rho_1}{h_1}{A_1}{\Gamma_1,\Delta_1,\Theta_1}$, we know that
  $y.\pi$ will lead to a value $v$ (formally,
  $\evalexpr{\rho_1}{h_1}{y.\pi}{v}$) and a node $n$ (formally,
  $\tevalexpr{\Gamma_1}{\Delta_1}{y.\pi}{n}$) such that
  ${\Interpret{\rho_1}{h_1}{A_1}{\Gamma_1}{\Delta_1}{v}{n}}$. The two paths
  being identical save for their prefix, this property also holds for $x.\pi$
  (formally,   ${\Interpret{\rho_1[x\mapsto\rho_1(y)]}{h_1}{A_1}{\Gamma_1[x\mapsto\Gamma_1(y)]}{\Delta_1}{v}{n}}$,
  $\evalexpr{\rho_1[x\mapsto\rho_1(y)]}{h_1}{x.\pi}{v}$,
  and $\tevalexpr{\Gamma_1[x\mapsto\Gamma_1(y)]}{\Delta_1}{x.\pi}{n}$).
  Paths not beginning with $x$ are not affected by the assignment, and so we
  can conclude that the first premise is satisfied. 

  For the second premise, let $n \in \Theta_2$ and assume that $n$ can
  be reached by two paths $\pi$ and $\pi'$ in $\Delta_2$. If none or
  both of the paths begin with $x$ then, by assumption, the two paths
  will lead to the same location in $h_2=h_1$. Otherwise, suppose
  that, say, $\pi$ begins with $x$ and $\pi'$ with a different
  variable $z$ and that they lead to $l$ and $l'$ respectively. Since
  $\Gamma_2(x)=\Gamma_1(y)$ and $\Delta_2=\Delta_1$, then by assumption there is a path $\pi''$ in
  $(\rho_1,h_1,\Delta_1,\Gamma_1)$ that starts with $y$, and such that
  $\evalexpr{\rho_1}{h_1}{\pi''}{l}$. As $z$ is not affected by the assignment,
  we also have that $\evalexpr{\rho_1}{h_1}{\pi'}{l'}$.  Therefore, as $n \in
  \Theta_1$ and $\InterpMMM{\rho_1}{h_1}{A_1}{\Gamma_1,\Delta_1,\Theta_1}$, we
  can conclude that $l = l'$.  This proves that
  $\InterpMMM{\rho_1[x\mapsto\rho_1(y)]}{h_1}{A_1}{\Gamma_1[\var{x}\mapsto\Gamma_1(\var{y})],\Delta_1,\Theta_1}$. 

  \underline{If $c\equiv \HAssign{x.f}{y}$}
  then $\st{\rho_2,h_2,A_2} = \st{\rho_1,h_1[(\rho_1(x),f)\mapsto \rho_1(y)],A_1}$.
  Two cases arise, depending on whether $n = \Gamma(x)$ is a strong or a weak node.

  \underline{If $c\equiv \HAssign{x.f}{y}$ and $n=\Gamma(x)\in\Theta$} 
    the node $n$ represents a unique concrete cell in $h$. To check the first
    premise, we make use of Lemma~\ref{lem:strongassign} on a given path
    $\pi$, a node $n$ and a location $l$ such that
    $\tevalexpr{\Gamma_2}{\Delta_2}{\pi}{n}$ and
    $\evalexpr{\rho_2}{h_2}{\pi}{l}$. This yields one of two main hypotheses.
    In the first case $\pi$ is not modified by the $f$-redirection (formally,
    $\evalexpr{\rho_1}{h_1}{\pi}{l} \wedge
    \tevalexpr{\Gamma_1}{\Delta_1}{\pi}{n}$), and by assumption
    $\Interpret{\rho_1}{h_1}{A_1}{\Gamma_1}{\Delta_1}{l}{n}$. Now consider any
    path $\pi_0$ such that $\evalexpr{\rho_2}{h_2}{\pi_0}{l}$: there is a
    node $n_0$ such that $\tevalexpr{\rho_2}{h_2}{\pi_0}{n_0}$, and we can
    reapply Lemma~\ref{lem:strongassign} to find that
    $\Interpret{\rho_1}{h_1}{A_1}{\Gamma_1}{\Delta_1}{l}{n_0}$. Hence $n_0=n$,
    and since nodes in $\Delta_1$ and $\Delta_2$ are untouched by the field
    assignment typing rule, we can conclude that the first premise is
    satisfied. In the second case ($\pi$ is modified by the $f$-redirection)
    there is a path $\pi'$ in $(\rho_1,h_1,\Gamma_1,\Delta_1)$ such that
    $y.\pi'$ leads respectively to $l$ and $n$ (formally
    $\evalexpr{\rho_1}{h_1}{y.\pi'}{l} \wedge
    \tevalexpr{\Gamma_1}{\Delta_1}{y.\pi'}{n}$). 
    As in the previous case, for
    any path $\pi_0$ such that $\evalexpr{\rho_2}{h_2}{\pi_0}{l}$ and
    $\tevalexpr{\rho_2}{h_2}{\pi_0}{n_0}$, we have
    $\Interpret{\rho_1}{h_1}{A_1}{\Gamma_1}{\Delta_1}{l}{n_0}$, which implies
    that $n=n_0$, and thus we can conclude 
    that $\pi_0$ leads to the same node as $\pi$, as required.

    For the second premise, let $n \in \Theta_2$ and assume that $n$ can be
    reached by two paths $\pi$ and $\pi'$ in $\Delta_2$. The application of
    Lemma~\ref{lem:strongassign} to both of these paths yields the following
    combination of cases:
    \begin{desCription}
      \item\noindent{\hskip-12 pt\bf neither path is modified by the $f$-redirection:}\ formally,
      $\evalexpr{\rho_1}{h_1}{\pi}{l} \wedge \tevalexpr{\Gamma_1}{\Delta_1}{\pi}{n} \wedge
      \evalexpr{\rho_1}{h_1}{\pi'}{l'} \wedge \tevalexpr{\Gamma_1}{\Delta_1}{\pi'}{n}$.
      By assumption, $l=l'$.
      \item\noindent{\hskip-12 pt\bf one of the paths is modified by the $f$-redirection:}\ without loss
      of generality, assume $\evalexpr{\rho_1}{h_1}{\pi'}{l'} \wedge
      \tevalexpr{\Gamma_1}{\Delta_1}{\pi}{n}$, and there is a path $\pi_\star$ in
      $(\rho_1,h_1,\Gamma_1,\Delta_1)$ such that $y.\pi_\star$ leads to $l$
      in the heap, and $n$ in the graph (formally,
      $\evalexpr{\rho_1}{h_1}{y.\pi_\star}{l} \wedge
      \tevalexpr{\Gamma_1}{\Delta_1}{y.\pi_\star}{n}$). By assumption, $l=l'$.
      \item\noindent{\hskip-12 pt\bf both paths are modified by the $f$-redirection:}\ we can find two
      paths $\pi_\star$ and $\pi'_\star$ such that $y.\pi_\star$ leads to $l$
      in the heap and $n$ in the graph, and $y.\pi'_\star$ leads to $l'$ in
      the heap and $n$ in the graph (formally, $\evalexpr{\rho_1}{h_1}{y.\pi_\star}{l} \wedge
      \tevalexpr{\Gamma_1}{\Delta_1}{y.\pi_\star}{n} \wedge
      \evalexpr{\rho_1}{h_1}{y.\pi'_\star}{l'} \wedge
      \tevalexpr{\Gamma_1}{\Delta_1}{y.\pi'_\star}{n}$). By assumption,
      $l=l'$.
    \end{desCription}
    In all combinations, $l=l'$ in $h_1$. Since the rule for field assignment in
    the operational semantics preserves the locations, then $l=l'$ in $h_2$.

    This concludes the proof of $\InterpMMM{\rho_1}{h_1[(\rho_1(x),f)\mapsto
    \rho_1(y)]}{A_1}{\Gamma_1,\Delta_1[n,f \mapsto \Gamma(y)],\Theta_1}$ when
    $n\in\Theta$.

  \underline{If $c\equiv \HAssign{x.f}{y}$ and $n=\Gamma(x)\notin\Theta$} 
    here $n$ may represent multiple concrete cells in $h$. Let $\sigma_1$ and
    $\sigma_2$ be the mappings defined, respectively, by the hypothesis
    $(\Gamma_1,\Delta_1,\Theta_1)\sqsubseteq(\Gamma_2,\Delta_2,\Theta_2)$
    and
    $(\Gamma_1,\Delta_1[n,f\mapsto
    \Gamma_1(y)],\Theta_1)\sqsubseteq(\Gamma_2,\Delta_2,\Theta_2)$.
    The first premise of the proof is proved by examining a fixed path $\pi$
    in $(\rho_2,h_2,\Gamma_2,\Delta_2)$ that ends in $l_0$ in the concrete
    heap, and $n_0'$ in the abstract graph. Applying Lemma~\ref{lem:pathdec}
    to this path (formally, instantiating $l$ by $\rho_1(x)$, $l'$ by
    $\rho_1(y)$, $l_f$ by $l_0$, $\pi$ by $y$, and $\pi'$ by $\pi$) yields two
    possibilities. The \emph{first alternative} is when $\pi$ is not modified by the
    $f$-redirection (formally, $\evalexpr{\rho_1}{h_1}{\pi}{l_0}$).
    Lemma~\ref{lem:subtype} then asserts the existence of a node $n_0$ that
    $\pi$ evaluates to in $(\rho_1,h_1,\Gamma_1,\Delta_1)$ (formally,
    $\tevalexpr{\Gamma_1}{\Delta_1}{\pi}{n_0}$ with
    $n'_0=\sigma_1(n_0)$). 
    Moreover, by assumption $n_0$ and $l_0$ are in correspondence (formally,
    $\Interpret{\rho_1}{h_1}{A_1}{\Gamma_1}{\Delta_1}{l_0}{n_0}$). To prove
    that $l_0$ and $n'_0$ are in correspondence, we refer to the auxiliary
    type interpretation rule in Fig.~\ref{fig:type:interpret}, and prove that
    given a path $\pi_0$ that verifies $\evalexpr{\rho_2}{h_2}{\pi_0}{l_0}$,
    the proposition $\tevalexpr{\Gamma_2}{\Delta_2}{\pi_0}{n'_0}$ holds. 
    Using Lemma~\ref{lem:pathdec} on $\pi_0$ (formally, instantiating $l$ by
    $\rho_1(x)$, $l'$ by $\rho_1(y)$, $l_f$ by $l_0$, $\pi$ by $y$, and $\pi'$
    by $\pi_0$), the only non-immediate case is when $\pi_0$  goes through
    $f$. In this case, $\pi_0 = \pi_z.f.\pi_1.f.\ldots.f.\pi_n.f.\pi_f$, and
    we can reconstruct this as a path to $n'_0$ in $(\Gamma_2,\Delta_2)$ by
    assuming there are two nodes $n_x'$ and $n_y'$ such that
    $\sigma_1(\Gamma_1(x)) = n_x'$ and $\sigma_1(\Gamma_1(y)) = n_y'$, and
    observing:

    \begin{enumerate}[$\bullet$]
    \item $\evalexpr{\rho_1}{h_1}{\pi_z}{\rho_1(x)}$ thus
    $\tevalexpr{\Gamma_1}{\Delta_1}{\pi_z}{\Gamma_1(x)}$ by assumption.
    Because access path evaluation is monotonic wrt. mappings (a direct
    consequence of clause~\eqref{eq:st2} in the definition of sub-typing), we
    can derive $\tevalexpr{\Gamma_2}{\Delta_2}{\pi_z}{n_x'}$;
    \item for $i\in [1,n]$, $\evalexpr{\rho_1}{h_1}{y.\pi_i}{\rho_1(x)}$ thus
    by assumption $\tevalexpr{\Gamma_1}{\Delta_1}{y.\pi_i}{\Gamma_1(x)}$. By
    again using the monotony of access path evaluation, we can derive
    $\tevalexpr{\Gamma_2}{\Delta_2}{y.\pi_i}{n_x'}$;
    \item $\evalexpr{\rho_1}{h_1}{y.\pi_f}{l_0}$ thus by assumption
    $\tevalexpr{\Gamma_1}{\Delta_1}{y.\pi_f}{n_0}$. Hence
    $\tevalexpr{\Gamma_2}{\Delta_2}{y.\pi_f}{n'_0}$ by $n'_0 = \sigma_1(n_0)$
    and due to monotonicity.
    \item in $\Delta_1[n,f \mapsto \Gamma_1(y)]$, $n = \Gamma_1(x)$ points to
    $\Gamma_1(y)$ by $f$. Note that because the types
    $\Gamma_1,\Delta_1,\Theta_1$ and $\Gamma_1,\Delta_1[n,f\mapsto
    \Gamma_1(y)],\Theta_1$ share the environment $\Gamma_1$, we have
    $\sigma_2(\Gamma_1(x)) = \sigma_1(\Gamma_1(x)) = n_x'$ and
    $\sigma_2(\Gamma_1(y)) = \sigma_1(\Gamma_1(y)) = n_y'$.
    Hence in $\Delta_2$, $n_x'$ points to $n_y'$ by $f$, thanks to the
    monotonicity of $\sigma_2$.
    \end{enumerate}
    This concludes the proof of $\tevalexpr{\Gamma_2}{\Delta_2}{\pi_0}{n'_0}$.
    The cases when $n_x'$ and $n_y'$ do not exist are easily dismissed;
    we refer the reader to the Coq development for more details.

    We now go back to our first application of Lemma~\ref{lem:pathdec} and
    tackle the \emph{second alternative} -- when $\pi$ is indeed modified by
    the $f$-redirection. Here $\pi$ can be decomposed into
    $\pi_z.f.\pi_1.f.\ldots.f.\pi_n.f.\pi_f$. We first find the node in
    $\Delta_1$ that corresponds to the location $l_0$: for this we need to
    find a path in $(\rho_1,h_1,\Gamma_2,\Delta_2)$ that leads to both
    $l_0$ and $n'_0$ (in order to apply Lemma~\ref{lem:subtype}, we can no
    longer assume, as in the first alternative, that $\pi$ leads to $l_0$ in
    $\st{\rho_1,h_1}$). Since $\evalexpr{\rho_1}{h_1}{y.\pi_f}{l_0}$, the path
    $y.\pi_f$ might be a good candidate: we prove that it points to $n'_0$
    in $\Delta_2$.  Assume the existence of two nodes $n_x'$ and $n_y'$ in
    $\Delta_2$ such that $n_x' = \sigma_1(\Gamma_1(x)) =
    \sigma_2(\Gamma_1(x))$ and $n_y' = \sigma_1(\Gamma_1(y)) =
    \sigma_2(\Gamma_1(y))$ (the occurrences of non-existence are dismissed by
    reducing them, respectively, to the contradictory case when
    $\Gamma_2(x)=\top$, and to the trivial case when $n'_0=\top$; the equality
    between results of $\sigma_1$ and $\sigma_2$ stems from the same reasons
    as in bullet $4$ previously). From the first part of the decomposition of
    $\pi$ one can derive, by assumption, that $\pi_z$ leads to $\Gamma_1(x)$
    in $\Delta_1$, and, by monotonicity, to $n_x'$ in $\Delta_2$ (formally,
    $\evalexpr{\rho_1}{h_1}{\pi_z}{\rho_1(x)}$ entails
    $\tevalexpr{\Gamma_1}{\Delta_1}{\pi_z}{\Gamma_1(x)}$ implies
    $\tevalexpr{\Gamma_2}{\Delta_2}{\pi_z}{n_x'}$). Similarly, we can derive
    that for any $i\in[1,n]$,
    $\tevalexpr{\Gamma_1}{\Delta_1}{y.\pi_i}{\Gamma_1(x)}$. 
    We use monotonicity both to infer that in $\Delta_2$, $n_x'$ points to $n_y'$
    by $f$, and that for any $i\in[1,n]$
    $\tevalexpr{\Gamma_2}{\Delta_2}{y.\pi_i}{n_x'}$. From these three
    statements on $\Delta_2$, we have
    $\tevalexpr{\Gamma_2}{\Delta_2}{\pi_z.f.\pi_1.f.\ldots.f.\pi_n}{n_x'}$,
    and by path decomposition $\tevalexpr{\Gamma_2}{\Delta_2}{y.\pi_f}{n_0'}$.
    This allows us to proceed and apply Lemma~\ref{lem:subtype}, asserting the
    existence of a node $n_0$ in ($\rho_1,h_1,\Delta_1,\Gamma_1$) that
    $y.\pi_f$ evaluates to (formally,
    $\tevalexpr{\Gamma_1}{\Delta_1}{y.\pi_f}{n_0}$ with $n_0' =
    \sigma_1(n_0)$). Moreover, by assumption $n_0$ and $l_0$ are in
    correspondence (formally,
    $\Interpret{\rho_1}{h_1}{A_1}{\Gamma_1}{\Delta_1}{l_0}{n_0}$). The proof
    schema from here on is quite similar to what was done for the first
    alternative. As previously, we demonstrate that $l_0$ and $n'_0$ are in
    correspondence by taking a path $\pi_0$ such that
    $\evalexpr{\rho_2}{h_2}{\pi_0}{l_0}$ and proving
    $\tevalexpr{\Gamma_2}{\Delta_2}{\pi_0}{n'_0}$. As previously, using
    Lemma~\ref{lem:pathdec} on $\pi_0$ (formally, instantiating $l$ by
    $\rho_1(x)$, $l'$ by $\rho_1(y)$, $l_f$ by $l_0$, $\pi$ by $y$, and $\pi'$
    by $\pi_0$), the only non-immediate case is when $\pi_0$  goes through
    $f$. In this case, $\pi_0 = \pi'_z.f.\pi'_1.f.\ldots.f.\pi'_n.f.\pi'_f$,
    and we can reconstruct this as a path to $n'_0$ in $(\Gamma_2,\Delta_2)$
    by observing:
    \begin{enumerate}[$\bullet$]
    \item $\evalexpr{\rho_1}{h_1}{\pi'_z}{\rho_1(x)}$ thus
    $\tevalexpr{\Gamma_1}{\Delta_1}{\pi'_z}{\Gamma_1(x)}$ by assumption.
    Monotonicity yields $\tevalexpr{\Gamma_2}{\Delta_2}{\pi'_z}{n_x'}$;
    \item for $i\in [1,n]$, $\evalexpr{\rho_1}{h_1}{y.\pi'_i}{\rho_1(x)}$ thus
    by assumption $\tevalexpr{\Gamma_1}{\Delta_1}{y.\pi'_i}{\Gamma_1(x)}$, and
    by monotonicity we can derive $\tevalexpr{\Gamma_2}{\Delta_2}{y.\pi'_i}{n_x'}$;
    \item $\evalexpr{\rho_1}{h_1}{y.\pi_f}{l_0}$ thus since $l_0$ and $n_0$
    are in correspondence, $\tevalexpr{\Gamma_1}{\Delta_1}{y.\pi_f}{n_0}$.
    Hence $\tevalexpr{\Gamma_2}{\Delta_2}{y.\pi_f}{n'_0}$ by $n'_0 =
    \sigma_1(n_0)$ and by monotonicity. 
    \end{enumerate}
    This concludes the proof of $\tevalexpr{\Gamma_2}{\Delta_2}{\pi_0}{n'_0}$,
    hence the demonstration of the first premise.

    We are now left with the second premise, stating the unicity of strong
    node representation. Assume $n\in\Theta_2$ can be reached by two paths
    $\pi$ and $\pi'$ in $\rho_2,h_2,\Gamma_2,\Delta_2$, we use
    Lemma~\ref{lem:pathdec} to decompose both paths. The following cases
    arise:
    \begin  {desCription}
      \item\noindent{\hskip-12 pt\bf neither path is modified by the $f$-redirection:}\ formally,
      $\evalexpr{\rho_1}{h_1}{\pi}{l} \wedge
      \evalexpr{\rho_1}{h_1}{\pi'}{l'}$. Lemma~\ref{lem:subtype} ensures the
      existence of $n'$ such that $n = \sigma_1(n')$ and
      $\tevalexpr{\Gamma_1}{\Delta_1}{\pi}{n'} \wedge
      \tevalexpr{\Gamma_1}{\Delta_1}{\pi'}{n'}$ Thus by assumption, $l=l'$.
      \item\noindent{\hskip-12 pt\bf one of the paths is modified by the $f$-redirection:}\ without loss
      of generality, assume $\evalexpr{\rho_1}{h_1}{\pi}{l}$, and
      $\pi'=\pi_0.f.\pi_1.f.\ldots.f.\pi_n.f.\pi_f$ with
      $\evalexpr{\rho_1}{h_1}{y.\pi_f}{l'}$. From the first concrete path
      expression, Lemma~\ref{lem:subtype} ensures that there is $n'$ such that
      $n=\sigma_1(n')$, and $\tevalexpr{\Gamma_1}{\Delta_1}{\pi}{n'}$.
      Moreover, by assumption there exists a strong node $n''$ in $\Delta_1$
      that $y.\pi_f$ leads to, and that is mapped to $n$ by $\sigma_1$
      (formally, $n=\sigma_1(n'')$ and
      $\tevalexpr{\Gamma_1}{\Delta_1}{y.\pi_f}{n''}$).  Clause~\eqref{eq:st3}
      of the definition of subtyping states the uniqueness of strong source
      nodes: thus $n'=n''$, and we conclude by assumption that $l=l'$.
      \item\noindent{\hskip-12 pt\bf both paths are modified by the $f$-redirection:}\
      Lemma~\ref{lem:pathdec} provides two paths $\pi_f$ and $\pi'_f$ such
      that $y.\pi_f$ leads to $l$, and $y.\pi'_f$ leads to $l'$ in the heap
      (formally, $\evalexpr{\rho_1}{h_1}{y.\pi_\star}{l} \wedge
      \evalexpr{\rho_1}{h_1}{y.\pi'_\star}{l'}$). As in the previous case, we
      use the assumption and clause~\eqref{eq:st3} of main sub-typing
      definition to exhibit the common node $n'$ in $\Delta_1$ that is
      pointed to by both paths. By assumption, $l=l'$.
    \end    {desCription}
    This concludes the demonstration of the second premise, and of the
    $c\equiv \HAssign{x.f}{y}$ case in our global induction.
\qed
    %


\begin{thm}[Type soundness]
If ${ \vdash p}$ then
all methods $m$ declared in the program $p$
are secure, \emph{i.e.}, respect their copy policy.
\end{thm}
\begin{CA}
  For the proofs see~\cite{JensenKP10} and the companion Coq development.
\end{CA}
\begin{TR}
  \begin{proof}

[See Coq proof \texttt{Soundness.soundness}~\cite{clone-webpage}]
Given a method $m$ and a copy signature $\Copy(X)\{\tau\}$ attached to it, we show that 
 for all heaps $h_1,h_2\in\Heap$, local environments
$\rho_1,\rho_2\in\Env$, locally allocated locations $A_1,A_2\in\Power(\Loc) $, and variables $x,y\in\Var$,
$$
(\HCall{x}{\var{cn}:X}{y}, \st{\rho_1,h_1,A_1}) \leadsto \st{\rho_2,h_2,A_2}
  ~~\text{implies}~~\rho_2,h_2,x \models \tau.$$

Following the rule defining the semantics of calls to copy methods, we
consider the situation where there exists a potential overriding
of $m$,  
$\Copy(X')~\HAssign{m(a)}c$ such that
$$(c,\st{\rho_{\vnull}[a\mapsto\rho_1(y)],h_1,\emptyset}) \leadsto \st{\rho',h_2,A'}.$$
Writing $\tau'$ for the policy attached to $X'$, we know that $\tau \subseteq \tau'$.
By Lemma~\ref{lem:mono}, it is then sufficient to prove that  
$\rho_2,h_2,x \models \tau'$ holds.

By typability of $\Copy(X')~\HAssign{m(a)}c$, there exist $\Gamma',\Gamma,\Delta,\Theta$ such that 
$\Gamma'(\var{ret}) =  n_\tau$, 
$(\Gamma, \Delta, \Theta) \sqsubseteq (\Gamma',\Delta_\tau,\{n_\tau\})$ and
 $[~\cdot\mapsto\bot][x\mapsto \topout], \emptyset, \emptyset \vdash c : \Gamma,\Delta,\Theta$.

Using Theorem~\ref{theo2}, we know that 
$\InterpMMM{\rho'}{h_2}{A'}{\Gamma,\Delta,\Theta}$ holds and
by subtyping between 
$(\Gamma, \Delta, \Theta)$ and $(\Gamma',\Delta_\tau,\{n_\tau\})$ we 
obtain that $\InterpMMM{\rho'}{h_2}{A'}{\Gamma',\Delta_\tau,\{n_\tau\}}$ holds.
Theorem~\ref{th:policy-type} then immediately yields that $\rho_2,h_2,x \models \tau$.
  \end{proof}
\end{TR}

\section{Inference}\label{sec:inference}

In order to type-check a method with the previous type system, it is necessary
to infer intermediate types at each loop header, conditional junction
points and weak field assignment point. A standard approach consists in turning the previous typing problem
into a fixpoint problem in a suitable sup-semi-lattice structure. This section
presents the lattice that we put on $(\Type,\sqsubseteq)$. Proofs are
generally omitted by lack of space but can be found in the companion report.
Typability is then checked by computing a suitable least-fixpoint in this
lattice. We end this section by proposing a widening operator that is
necessary to prevent infinite iterations.


\begin{TR}
  \begin{lem}
  The binary relation $\sqsubseteq$ is a preorder on $\Type$.
  \end{lem}
   \begin{proof}
   The relation is reflexive because for all type $T\in\Type$, $T\sqsubseteq_\textit{id} T$.
   The relation is transitive because if there exists types $T_1, T_2, T_3\in\Type$
   such that $T_1\sqsubseteq_{\sigma_1}T_2$ and $T_2\sqsubseteq_{\sigma_2}T_3$
   for some fusion maps $\sigma_1$ and $\sigma_2$ then $T_1\sqsubseteq_{\sigma_3}T_3$
   for the fusion map $\sigma_3$ define by $\sigma_3(n) = \sigma_2(\sigma_1(n))$
   if $\sigma_1(n)\in\Node$ or $\top$ otherwise.
   \end{proof}
\end{TR}

We write $\equiv$ for the equivalence relation defined by $T_1\equiv
T_2$ if and only if $T_1\sqsubseteq T_2$ and $T_2\sqsubseteq
T_1$. Although this entails that $\sqsubseteq$ is a partial order structure on top of
$(\Type,\equiv)$, equality and order testing remains
difficult using only this definition.  Instead of considering the
quotient of $\Type$ with $\equiv$, we define a notion of
\emph{well-formed} types on which $\sqsubseteq$ is antisymmetric.  To
do this, we assume that the set of nodes, variable names and field
names are countable sets and we write $n_i$ (resp. $x_i$ and $f_i$)
for the $i$th node (resp. variable and field).  A type
$(\Gamma,\Delta,\Theta)$ is \emph{well-formed} if every node in
$\Delta$ is reachable from a node in $\Gamma$ and the nodes in
$\Delta$ follow a canonical numbering based on a breadth-first
traversal of the graph. 
Any type can be \emph{garbage-collected} into a canonical well-formed
type by removing all unreachable nodes from variables 
and renaming all remaining nodes using a fixed
strategy based on a total ordering on variable names and field names
and a breadth-first traversal. We call this transformation  $\gc$. The
following example shows the effect of $\gc$ using a canonical numbering.
\begin{center}
\includegraphics[width=.7\textwidth]{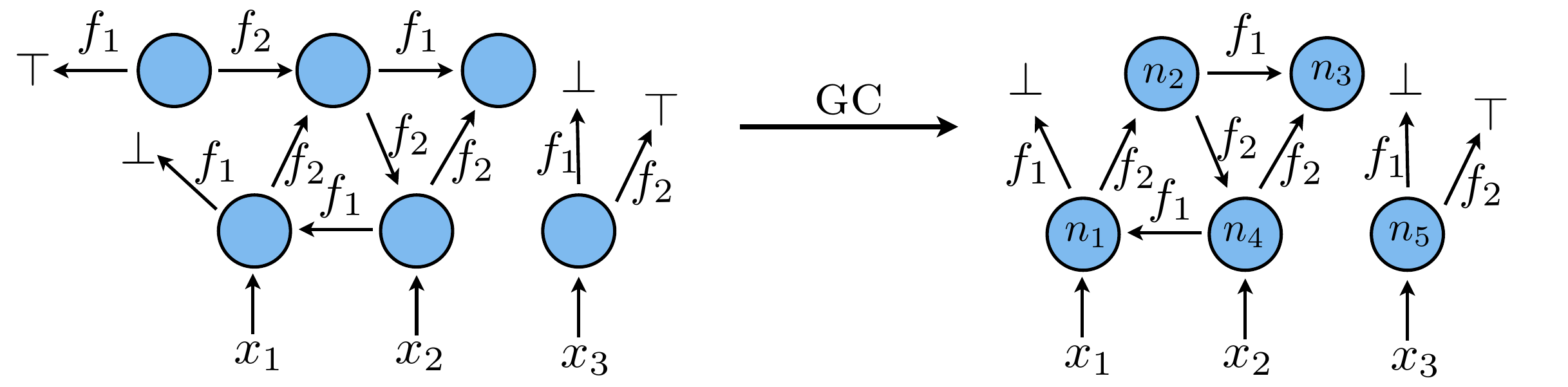}
\end{center}

\noindent
Since by definition, $\sqsubseteq$ only deals with reachable nodes, the $\gc$
function is a $\equiv$-morphism and respects type interpretation.
This means that an inference engine can at any time replace a type by a garbage-collected
version. This is useful to perform an equivalence test in order to
check whether a fixpoint iteration has ended.
\begin{lem}
For all well-formed types $T_1,T_2\in\Type$, 
$T_1\equiv T_2~~\text{iff}~~ T_1=T_2$.
\end{lem}


\begin{defi}
Let $\sqcup$ be an operator that merges two types according to the algorithm
in Fig.~\ref{fig:join:code}.
\end{defi}

\begin{figure}
  \begin{minipage}[c]{0.48\linewidth}
  \lstset{basicstyle=\ttfamily\scriptsize\color{black},float,emph={fusion,succ,updatemap,lift,ground,elementOf},emphstyle=\underbar,mathescape=true}
  \begin{lstlisting}
// Initialization.
// $\alpha$-nodes are sets in $\BaseType$.
// $\alpha$-transitions can be 
//    non-deterministic.
$\alpha$ = lift($\Gamma_1$,$\Gamma_2$,$\Delta_1 \cup \Delta_2$) 

// Determinize $\alpha$-transitions:
//    start with the entry points.
for $\{(x,t);(x,t')\} \subseteq (\Gamma_1\times\Gamma_2)$ {
  fusion($\{t,t'\}$)
}

// Determinize $\alpha$-transitions:
//    propagate inside the graph.
while $\exists u\in\alpha,\exists f\in\Field, |succ(u,f)| > 1$ {
  fusion(succ($u$,$f$))
}

// $\alpha$ is now fully determinized:
//    convert it back into a type.
$(\Gamma,\Delta,\Theta)$ = ground($\Gamma_1$,$\Gamma_2$,$\alpha$)
  \end{lstlisting}
  \end{minipage}
  \begin{minipage}[c]{0.48\linewidth}
  \lstset{basicstyle=\ttfamily\scriptsize\color{black},float,emph={fusion,succ,updatemap,lift,ground,elementOf},emphstyle=\underbar,mathescape=true}
  \begin{lstlisting}
// $S$ is a set of $t\in\BaseType$.
// $\llparenthesis S\rrparenthesis$ denotes a node labelled by $S$.
void fusion ($S$) { 
  // Create a new $\alpha$-node.
  NDG node n = $\llparenthesis S\rrparenthesis$
  $\alpha \gets \alpha + n$
  // Recreate all edges from the fused 
  //    nodes on the new node.
  for $t \in S$ {
    for $f \in \Field$ {
      if $\exists u, \alpha(t,f)=u$ {
        // Outbound edges.
        $\alpha \gets \alpha$ with $(n,f) \mapsto u$
      } 
      if $\exists n', \alpha(n',f)=t$ {
        // Inbound edges.
        $\alpha \gets \alpha$ with $(n',f) \mapsto n$
    }}
  // Delete the fused node.
  $\alpha \gets \alpha - t$
}}
  \end{lstlisting}
  \end{minipage}
\caption{Join Algorithm}
\label{fig:join:code}
\end{figure}

The procedure has $T_1=(\Gamma_1,\Delta_1,\Theta_1)$ and
$T_2=(\Gamma_2,\Delta_2,\Theta_2)$ as input, then takes the following steps.
\begin{enumerate}[(1)]
  \item \label{enum:cell} It first makes the disjunct union of $\Delta_1$ and
  $\Delta_2$ into a non-deterministic graph (\NDG) $\alpha$, where nodes are
  labelled by sets of elements in $\BaseType$. This operation is performed by the
  \ttt{lift} function, that maps nodes to singleton nodes, and fields to
  transitions.
  \item \label{enum:init} It joins together the nodes in $\alpha$ referenced
  by $\Gamma_i$ using the \ttt{fusion} algorithm\footnote{Remark that
  $\Gamma_i$-bindings are not represented in $\alpha$, but that node set
  fusions are trivially traceable. This allows us to safely ignore $\Gamma_i$
  during the following step and still perform a correct graph
  reconstruction.}.
  \item \label{enum:fus} Then it scans the \NDG and merges all
  nondeterministic successors of nodes.
  \item \label{enum:decell} Finally it uses the \ttt{ground} function to
  recreate a graph $\Delta$ from the now-\emph{deterministic} graph $\alpha$.
  This function operates by pushing a node set to a node labelled by the
  $\leq_\sigma$-sup of the set.  The result environment $\Gamma$ is derived
  from $\Gamma_i$ and $\alpha$ before the $\Delta$-reconstruction.
\end{enumerate}
All state fusions are recorded in a map $\sigma$ which binds nodes in
$\Delta_1\cup\Delta_2$ to nodes in $\Delta$. 
\begin{TR}
  Figure~\ref{fig:auxfun} contains the auxiliary functions used in the above
  procedure. Figure~\ref{fig:join:example} unfolds the algorithm on a small
  example.
\end{TR}

\begin{TR}
  \begin{figure}
    \begin{minipage}{0.48\linewidth}
    \lstset{basicstyle=\ttfamily\scriptsize\color{black},float,emph={fusion,succ,updatemap,lift,ground,elementOf},emphstyle=\underbar,mathescape=true}
    \begin{lstlisting}
    // Convert a type to an NDG
    NDG lift ($\Gamma_1$,$\Gamma_2$,$\Delta$) {
      NDG $\alpha$ = undef
      for $x \in \text{Var}$ {
        $\alpha \gets \alpha + \llparenthesis\{\Gamma_1(x)\}\rrparenthesis + \llparenthesis\{\Gamma_2(x)\}\rrparenthesis$
      }
      for $n \in \Delta$ {
        if $\exists f\in\text{Fields}, \exists b\in\text{BaseType}, \Delta[n,f]=b$ {
          $\alpha \gets \alpha + \llparenthesis\{n\}\rrparenthesis + \llparenthesis\{b\}\rrparenthesis$
          $\alpha \gets \alpha$ with $(\llparenthesis \{n\}\rrparenthesis,f) \mapsto \llparenthesis \{b\}\rrparenthesis$
      }}
      return $\alpha$
    }
    \end{lstlisting}
    \end{minipage}
    \begin{minipage}{0.48\linewidth}
    \lstset{basicstyle=\ttfamily\scriptsize\color{black},float,emph={fusion,succ,updatemap,lift,ground,elementOf},emphstyle=\underbar,mathescape=true}
    \begin{lstlisting}
    // Convert an NDG to a type
    $\Type$ ground ($\Gamma_1$,$\Gamma_2$,$\alpha$) {
      ($\Var\to\BaseType$) $\Gamma$ = $\lambda x.\bot$
      $\LSG$ $\Delta$ = undef
      for $N\in \alpha$ {
        for $x\in \Var, \Gamma_1(x)\in N \vee \Gamma_2(x)\in N$ {
          $\Gamma \gets \Gamma$ with $x \mapsto \nmlz{N}$
      }}
      for $N,N' \in \alpha$ {
        if $\alpha(N,f)=N'$ {
          $\Delta \gets \Delta$ with $(\nmlz{N},f) \mapsto \nmlz{N'}$
      }}
      return $(\Gamma,\Delta)$
    }
    \end{lstlisting}
    \end{minipage}
    \begin{minipage}[c]{0.40\linewidth}
    \lstset{basicstyle=\ttfamily\scriptsize\color{black},float,emph={fusion,succ,updatemap,lift,ground,elementOf},emphstyle=\underbar,mathescape=true}
    \begin{lstlisting}
    // $\leq_\sigma$-sup function
    BaseType $\downarrow$ (N) {
      if $\top\in N$ return $\top$
      else if $\forall c\in N, c=\bot$ return $\bot$
      else return freshNode()
    }
    \end{lstlisting}
  \end{minipage}
  \caption{Auxiliary join functions}
  \label{fig:auxfun}
  \end{figure}

  \begin{figure}
  \centering
  \includegraphics[width=.6\linewidth]{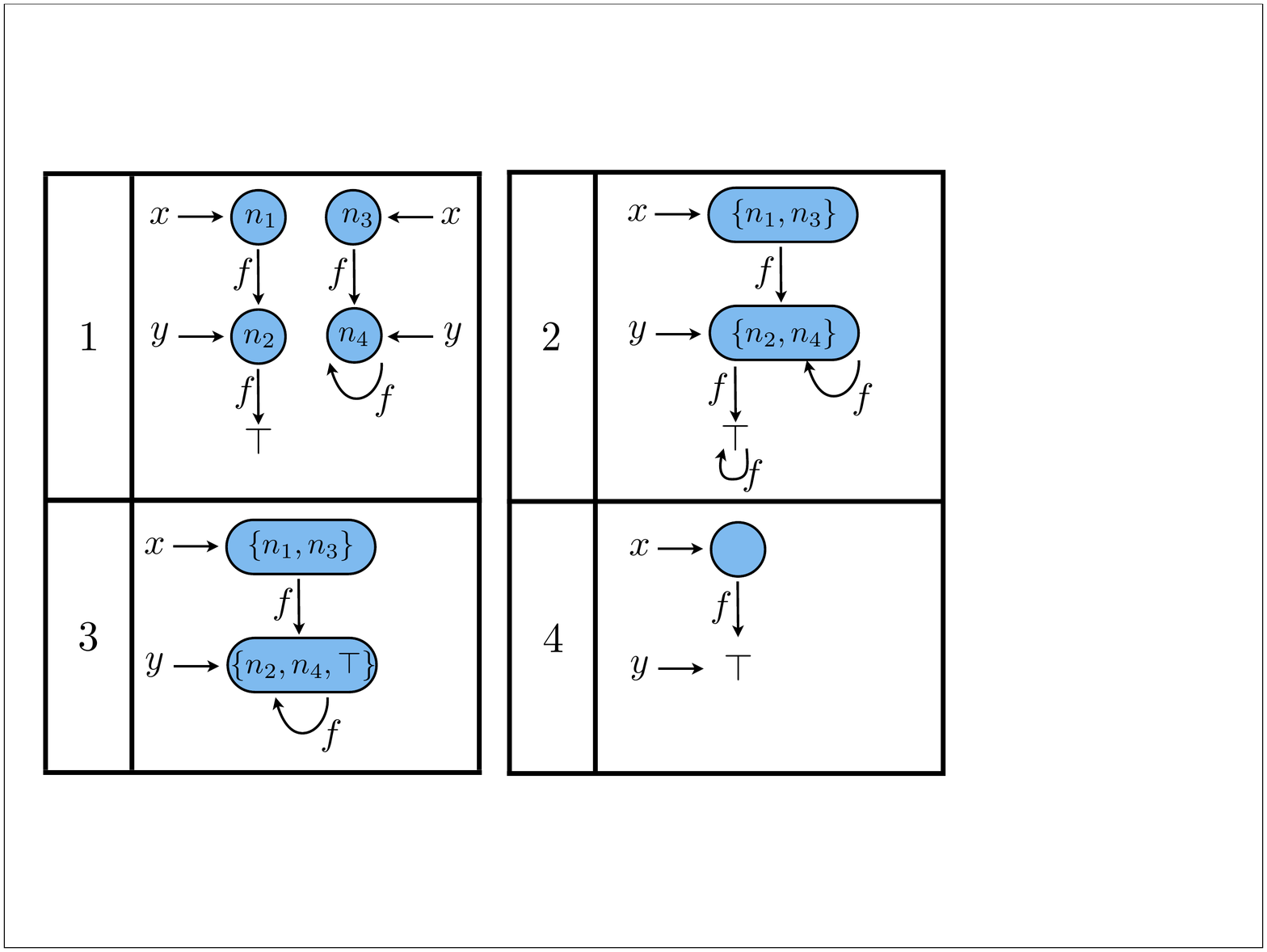}
    \caption{An example of join}
    \label{fig:join:example}
  \end{figure}
\end{TR}

\begin{thm}
The operator $\sqcup$ defines a sup-semi-lattice on types.
\end{thm}

\proof
\begin{CA}
  See~\cite{JensenKP10}.
\end{CA}
\begin{TR}
  First note that $\sigma_1$ and $\sigma_2$, the functions associated with the
  two respective sub-typing relations, can easily be reconstructed from
  $\sigma$. 
  \begin{desCription}
  \item\noindent{\hskip-12 pt\bf Upper bound:}\
  Let $(T,\sigma)= T_1\sqcup T_2$: we prove that $T_1 \sqsubseteq T$ and $T_2
  \sqsubseteq T$.
  Hypothesis~\eqref{eq:st2} is discharged by case analysis, on $t_1$ and on
  $t_2$. The general argument used is that the join algorithm does not
  delete any edges, thus preserving all paths in the initial graphs.
  \item\noindent{\hskip-12 pt\bf Least of the upper bounds:}\
  Let $(T,\sigma)= T_1 \sqcup T_2$. Assume there exists 
  $T'$ such that $T_1 \sqsubseteq T'$ and $T_2 \sqsubseteq T'$.  Then we prove
  that $T \sqsubseteq T'$.
  %
  The proof consists in checking that the join algorithm produces, at each
  step, an intermediary pseudo-type $T$ such that $T\sqsubseteq T'$.  The
  concrete nature of the algorithm drives the following, more detailed 
  decomposition.
  \newcommand{\alt}{\sqsubseteq_\tau}
  \newcommand{\alts}{\sqsubseteq_{\tau^{*}}}
  \begin{enumerate}[(1)]
    \item \label{pr:autorel} 
      Given a function $\sigma$, 
      define a state mapping function $\leq_\tau$, and a
      $\sqsubseteq_\sigma$-like relation $\alt$ on non-deterministic graphs.  
      The aim with this relation is to emulate the properties of
      $\sqsubseteq_\sigma$, lifting the partial order $\leq_\sigma$ on nodes to
      sets of nodes (cells).
      Lift $T'$ into an \NDG $\alpha'$.
    \item \label{pr:unioncel}
      Using the subtyping relations between $T_1$, $T_2$, and $T'$, establish
      that the disjunct union and join steps produce an intermediary
      \NDG $\beta$ such that $\beta \alt \alpha'$.
    \item \label{pr:fusion}
      Ensure that the fusions operated by the join algorithm in $\beta$
      produce an \NDG $\gamma$ such that $\gamma \alt \alpha'$. This
      is done by case analysis, and depending on whether the fusion takes place
      during the first entry point processing phase, or if it occurs later.
    \item \label{pr:decel}
      Using the $\BaseType$-lattice, show that the ground operation on the
      \NDG $\gamma$ produces a type $T$ that is a sub-type of $T'$.\qed
  \end{enumerate}
  \end{desCription}
\end{TR}

\noindent
The poset structure does have infinite ascending chains, as shown by
the following example. 
\begin{center}
\includegraphics[width=.9\textwidth]{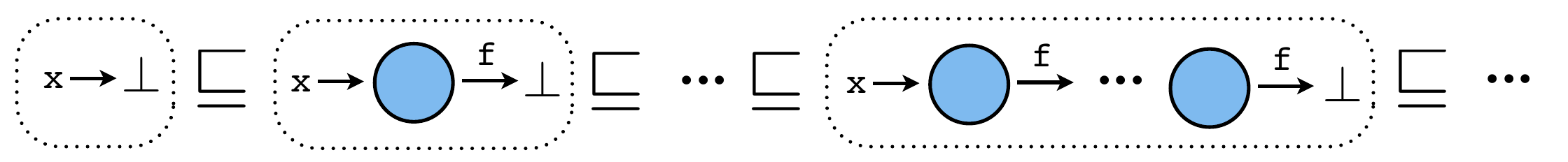}
\end{center}
Fixpoint iterations may potentially result in such an infinite chain
so we have then to rely on a widening~\cite{CousotCousot77} operator to enforce
termination of  fixpoint computations. Here we follow a pragmatic approach
and define a widening operator $\nabla\in\Type\times\Type\to\Type$ that takes
the result of $\sqcup$ and that collapses  together (with  the operator
\ttt{fusion} defined above) any node $n$ and its predecessors such that  the
minimal path reaching $n$ and starting from a local variable is of length at
least 2. 
\begin{TR}
This is illustrated by the following example.
\begin{center}
\includegraphics[width=.9\textwidth]{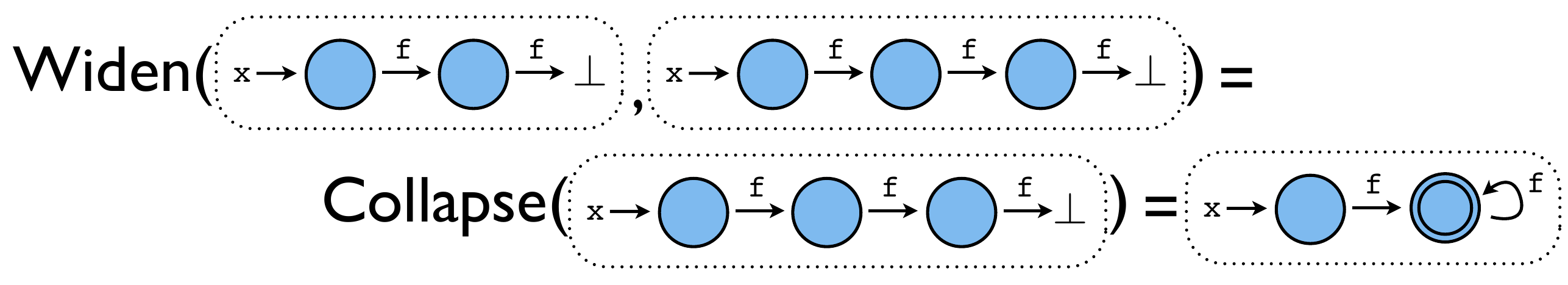}
\end{center}
This ensures the termination of the fixpoint iteration because the number of
nodes is then bounded by $2N$ with $N$ the number of local variables
in the program. 
\end{TR}

\section{Experiments}

The policy language and its enforcement mechanism have been implemented in the
form of a security tool for Java bytecode. Standard Java \ttt{@interface}
declarations are used to specify native annotations, which enable development
environments such as Eclipse to parse, identify and auto-complete
\ttt{@Shallow}, \ttt{@Deep}, and \ttt{@Copy} tags.  Source code annotations
are being made accessible to bytecode analysis frameworks. Both the policy
extraction and enforcement components are implemented using  the Javalib/Sawja
static analysis libraries\footnote{\url{http://sawja.inria.fr}} to derive
annotations and intermediate code representations, and to facilitate the
creation of an experimental Eclipse plugin.

In its standard mode, the tool performs a modular verification of
annotated classes.  We have run experiments on several
classes of the standard library (specially in the package \ttt{java.util})
and have successfully checked realistic copy signatures for them. 
These experiments have also
confirmed that the policy enforcement mechanism facilitates 
re-engineering into more compact implementations of cloning methods in
classes with complex dependencies, such as those forming the
\ttt{gnu.xml.transform} package. For example, in the \ttt{Stylesheet}
class an inlined implementation of multiple deep copy methods for half
a dozen fields can be rewritten to dispatch these functionalities to
the relevant classes, while retaining the expected copy policy. This
is made possible by the modularity of our enforcement mechanism, which
validates calls to external cloning methods as long as their
respective policies have been verified.  As expected, some cloning methods are
beyond the reach of the analysis. We have identified
one such method in GNU Classpath's \ttt{TreeMap} class, where the
merging of information at control flow merge points destroys too much
of the inferred type graph. A disjunctive form of abstraction seems
necessary to verify a deep copy annotation on such programs.



The analysis is also capable of processing un-annotated methods, albeit with
less precision than when copy policies are available---this is because it
cannot rely on annotations to infer external copy method types.
Nevertheless, this capability was used to test the tool on two large
code bases. The 17000 classes in Sun's \ttt{rt.jar} and the 7000
others in the GNU Classpath have passed our scanner
un-annotated. Among the 459 \ttt{clone()} methods we found in these
classes, only 15 have been rejected because of an assignment or method
call on non-local memory, as explained in
Section~\ref{sec:typesystem}. %
Assignment on non-local memory means here that the copying method is
updating fields of other objects than the result of the copy
itself. For such examples, our shape analysis seems too coarse to track
the dependencies between sources and copy targets.
 For 78 methods we were unable to infer 
the minimal, shallow signature $\{\}$ (the same signature as
\ttt{java.lang.Object.clone()}). In some cases, for instance in the
\ttt{DomAttr} class, this will happen when the copy method returns the result
of another, unannotated method call, and can be mitigated with additional copy
annotations. In other cases, merges between abstract values result in
precision losses: this is, for instance, the case for the \ttt{clone} method
of the TreeMap class, as explained above.

Our prototype confirms the efficiency of the enforcement technique:
these verifications took about 25s to run on stock hardware.
%
The prototype, the Coq formalization and proofs, as well as examples of annotated
classes can be found at \url{http://www.irisa.fr/celtique/ext/clones}.

%

%



\section{Related Work}
\label{sec:relatedwork}

Several proposals for programmer-oriented annotations of Java programs
have been published following Bloch's initial proposal of an annotation
framework for the Java language \cite{Bloch04}. 
These proposals define the syntax of the annotations but often leave
their exact semantics unspecified. A notable exception is the set of
annotations concerning non-null annotations \cite{Fahndrich03} for which a precise
semantic characterization has emerged~\cite{hubert08}. 
Concerning security, the GlassFish environment in Java offers program annotations of
members of a class (such
as \ttt{@DenyAll} or \ttt{@RolesAllowed}) for implementing role-based access control to
methods. 

To the best of our knowledge, the current paper is the first to
propose a formal, semantically founded framework for secure cloning through program
annotation and static enforcement. The closest work in
this area is that of Anderson \emph{et al.}~\cite{AndersonGayNaik:PLDI09}
who have designed an annotation system for C data structures in order
to control sharing between threads. Annotation policies are enforced by
a mix of static and run-time verification. On the run-time verification
side, their approach requires an operator that can dynamically
``cast'' a cell to an unshared structure. 
In contrast, our approach offers a completely static mechanism with
statically guaranteed alias properties.

Aiken \emph{et al.} proposes an analysis for checking and inferring
local non-aliasing of data~\cite{Aiken:03}. They propose to annotate C function parameters with
the keyword \ttt{restrict} to ensure that no other aliases to the data
referenced by the parameter are used during the execution of the
method. A type and effect system is defined for enforcing this
discipline statically. This analysis differs from ours in that it
allows aliases to exist as long as they are not used whereas we aim at
providing guarantees that certain parts of memory are without
aliases. 
The properties tracked by our type system are close to 
escape analysis~\cite{Blanchet99,ChoiGSSM99} but the analyses differ
in their purpose. While escape
analysis tracks locally allocated objects and tries to detect those
that do not escape after the end of a method execution, we are
specifically interested in tracking 
locally allocated objects that escape from the result of a method, as
well as analyse their dependencies with respect to parameters.



Our static enforcement technique falls within the large area of static
verification of heap properties. A substantial amount of research has been
conducted here, the most prominent being region calculus~\cite{TofteTalpin97},
separation logic~\cite{OHearnYR04} and shape analysis~\cite{SagivRW02}. Of
these three approaches, shape analysis comes closest in its use of shape
graphs.  Shape analysis is a large framework that allows to infer complex
properties on heap allocated data-structures like absence of dangling pointers
in C or non-cyclicity invariants. In this approach, heap cells are abstracted
by shape graphs with flexible object abstractions. Graph nodes can either
represent a single cell, hence allowing strong updates, or several cells
(summary nodes). \emph{Materialization} allows to split a summary node during
cell access in order to obtain a node pointing to a single cell.
The shape graphs that we use are not intended to do full shape analysis but
are rather specialized for tracking sharing in locally allocated objects. We
use a different naming strategy for graph nodes and discard all information
concerning non-locally allocated references. This leads to an analysis which
is more scalable than full shape analysis, yet still powerful enough
for verifying complex copy policies as demonstrated in the concrete case study
\ttt{java.util.LinkedList}.

Noble \emph{et al.}~\cite{Noble:98:Flexible} propose a prescriptive
technique for characterizing the aliasing, and more generally, the
topology of the object heap in object-oriented programs. This
technique is based on alias modes which have evolved into the notion
of ownership types \cite{Clarke:98:Ownership}. In this setting, the
annotation \ttt{@Repr} is used to specify that an object is
\emph{owned} by a specific object. It is called a
\emph{representation} of its owner. After such a declaration, the
programmer must manipulate the representation in order to ensure that
any access path to this object should pass trough its owner. Such a
property ensures that a \ttt{@Repr} field must be a \ttt{@Deep}
field in any copying method. Still, a \ttt{@Deep} field is not
necessarily a \ttt{@Repr} field since a copying method may want to
deeply clone this field without further interest in the global alias
around it.  Cloning seems not to have been studied further in the
ownership community and ownership type checking is generally not
adapted to flow-sensitive verification, as required by the programming
pattern exhibited in existing code. In this example, if we annotate
the field \ttt{next} and \ttt{previous} with \ttt{@Repr}, the
\ttt{clone} local variable will not be able to keep the same
ownership type at line 12 and at line 26. Such a an example would
require ownership type systems to track the update of a reference in
order to catch that any path to a representation has been erased in
the final result of the method.

We have aimed at annotations that together with static
analysis allows to verify existing cloning methods. 
Complementary to our approach, 
Drossopoulou and Noble~\cite{pubsdoc:clonesPre} propose a system
that generate cloning methods from annotation inpired by ownership types.

\section{Conclusions and Perspectives}

Cloning of objects is an important aspect of exchanging data with
untrusted code. Current language technology for cloning does not
provide adequate means for defining and enforcing a secure copy policy
statically; a task which is made more difficult by important
object-oriented features such as inheritance and re-definition of
cloning methods. We have presented a flow-sensitive type system for
statically enforcing copy policies defined by the software developer
through simple program annotations. The annotation formalism deals
with dynamic method dispatch and addresses some of the problems posed
by redefinition of cloning methods in inheritance-based object
oriented programming language (but see
Section~\ref{sec:pol:limitations} for a discussion of current
limitations). The verification technique is designed to enable modular
verification of individual classes. By specifically targeting the
verification of copy methods, we consider a problem for which it is
possible to deploy a localized version of shape analysis that avoids
the complexity of a full shape analysis framework. This means that our
method can form part of an extended, security-enhancing Java byte code
verifier which of course would have to address, in addition to secure cloning, a wealth of other
security policies and security guidelines as \emph{e.g.}, listed on
the CERT web site for secure Java
programming~\cite{CertGuidelines:2010}.

The present paper constitutes the formal foundations for a secure
cloning framework. All theorems except those of
Section~\ref{sec:inference} have been mechanized in the Coq proof
assistant. Mechanization was particularly challenging because
of the storeless nature of our type interpretation but in the end
proved to be  of great help to get the
soundness arguments right. 

 Several issues merit further investigations in order
to develop a full-fledged software security verification tool. 
The extension of the policy language to be able to impose policies on
fields defined in sub-classes should be developed
(\emph{cf.}~discussion in Section~\ref{sec:pol:limitations}). We
believe that the 
analysis defined in this article can be used to enforce such policies
but their precise semantics remains to be defined. 
In the current approach, virtual methods without copy
policy annotations are considered as black boxes that may modify any object reachable
from its arguments. An extension of our copy annotations to virtual
calls should be worked out if we want to enhance our enforcement
technique and accept more secure copying methods. More advanced verifications
will be possible if we
develop a richer form of type signatures for methods where the formal
parameters may occur in copy policies, in order to express a relation
between copy properties of returning objects and parameter fields.
The challenge here is to provide sufficiently expressive signatures
which at the same time remain humanly readable software contracts. 
The current formalisation has been developed for a sequential model of
Java. We conjecture that the extension to interleaving multi-threading
semantics is feasible and that it can be done without making major changes to the type system
because we only manipulate thread-local pointers. 

An other line of work could be to consider the correctness of
\ttt{equals()} methods with respect to copying methods, since we
generally expect \ttt{x.clone().equals(x)} to be \ttt{true}.  The
annotation system is already in good shape for such a work but a
static enforcement may require a major improvement of our specifically
tailored shape analysis.

\section{Acknowledgement}
We wish to thank the ESOP’11 and LMCS anonymous reviewers for their
helpful comments on this article. We specially thank the anonymous reviewer who suggested the 
\ttt{EvilList} example presented in Section~\ref{sec:pol:limitations}.

\bibliographystyle{plain}
\bibliography{bibli}

\begin{thebibliography}{10}

\bibitem{clone-webpage}
Secure cloning webpage.
\newblock \url{http://www.irisa.fr/celtique/ext/clones}, September 2011.

\bibitem{Aiken:03}
A.~Aiken, J.~S. Foster, J.~Kodumal, and T.~Terauchi.
\newblock Checking and inferring local non-aliasing.
\newblock In {\em Proc. of PLDI '03}, pages 129--140. ACM Press, 2003.

\bibitem{AndersonGayNaik:PLDI09}
Z.~Anderson, D.~Gay, and M.~Naik.
\newblock Lightweight annotations for controlling sharing in concurrent data
  structures.
\newblock In {\em Proc.~of PLDI'09}, pages 98--109. ACM Press, 2009.

\bibitem{Blanchet99}
B.~Blanchet.
\newblock Escape analysis for object-oriented languages: Application to {J}ava.
\newblock In {\em Proc. of OOPSLA}, pages 20--34. ACM Press, 1999.

\bibitem{Bloch04}
J.~Bloch.
\newblock {\em JSR 175: A metadata facility for the Java programming language}.
\newblock \url{http://jcp.org/en/jsr/detail?id=175}, September 30, 2004.

\bibitem{CertGuidelines:2010}
CERT.
\newblock {\em The CERT Sun Microsystems Secure Coding Standard for Java},
  2010.
\newblock \url{https://www.securecoding.cert.org}.

\bibitem{ChoiGSSM99}
J.D. Choi, M.~G., M.~J. Serrano, V.~C. Sreedhar, and S.~P. Midkiff.
\newblock Escape analysis for {J}ava.
\newblock In {\em Proc. of OOPSLA}, pages 1--19. ACM Press, 1999.

\bibitem{Clarke:98:Ownership}
D.~Clarke, J.~Potter, and J.~Noble.
\newblock Ownership types for flexible alias protection.
\newblock In {\em Proc.~of the 13th ACM SIGPLAN conference on Object-oriented
  programming, systems, languages, and applications}, OOPSLA '98, pages 48--64,
  New York, NY, USA, 1998. ACM.

\bibitem{CousotCousot77}
P{.} Cousot and R{.} Cousot.
\newblock Abstract interpretation: a unified lattice model for static analysis
  of programs by construction or approximation of fixpoints.
\newblock In {\em Proc. of POPL'77}, pages 238--252. ACM Press, 1977.

\bibitem{pubsdoc:clonesPre}
S.~Drossopoulou and J.~Noble.
\newblock {Trust the Clones}.
\newblock In {\em FoVEOOS 2011 - preproceedings}, September 2011.

\bibitem{Fahndrich03}
M.~F{\"a}hndrich and K.~R.~M. Leino.
\newblock Declaring and checking non-null types in an object-oriented language.
\newblock In {\em Proc.\ of OOPSLA'03}, pages 302--312. ACM Press, 2003.

\bibitem{hubert08}
L.~Hubert, T.~Jensen, and D.~Pichardie.
\newblock Semantic foundations and inference of non-null annotations.
\newblock In {\em Proc. of FMOODS'08}, volume 5051 of {\em LNCS}, pages
  132--149. Springer Berlin, 2008.

\bibitem{JensenKP:Esop11}
T.~Jensen, F.~Kirchner, and D.~Pichardie.
\newblock Secure the clones: Static enforcement of secure object copying.
\newblock In G.~Barthe, editor, {\em Proc.~of 20th European Symposium on
  Programming (ESOP 2011)}. Springer LNCS vol.~6602, 2011.

\bibitem{Noble:98:Flexible}
J.~Noble, J.~Potter, and J.~Vitek.
\newblock Flexible alias protection.
\newblock In {\em Proc.~of ECOOP'98}, pages 158--185. Springer LNCS vol.~1445,
  1998.

\bibitem{OHearnYR04}
P.~W. O'Hearn, H.~Yang, and J.~C. Reynolds.
\newblock Separation and information hiding.
\newblock In {\em Proc. of POPL'04}, pages 268--280. ACM Press, 2004.

\bibitem{SagivRW02}
S.~Sagiv, T.~W. Reps, and R.~Wilhelm.
\newblock Parametric shape analysis via 3-valued logic.
\newblock {\em ACM Trans. Program. Lang. Syst.}, 24(3):217--298, 2002.

\bibitem{SunGuidelines:2010}
Sun Developer Network.
\newblock {\em Secure Coding Guidelines for the Java Programming Language,
  version 3.0}, 2010.
\newblock \url{http://java.sun.com/security/seccodeguide.html}.

\bibitem{TofteTalpin97}
M.~Tofte and J.-P. Talpin.
\newblock Region-based memory management.
\newblock {\em Information and Computation}, 132(2):109--176, 1997.

\end{thebibliography}

\end{document}